\documentclass[twoside,11pt]{article}
\usepackage{jmlr2e}

\usepackage{lastpage}
\jmlrheading{19}{2018}{1-\pageref{LastPage}}{5/17; Revised
4/18}{9/18}{17-285}{Binyan Jiang, Xiangyu Wang, and Chenlei Leng}
\ShortHeadings{title}{Jiang, Wang, and Leng}

\RequirePackage{amsmath,graphicx,adjustbox}
\newtheorem{thm}{Theorem}
\newtheorem{cor}{Corollary}
\newtheorem{Lem}{Lemma}
\newtheorem{prop}{Proposition}

\DeclareMathOperator*{\esssup}{ess\,sup}
\DeclareMathOperator*{\essinf}{ess\,inf}
\newcommand{\argmin}{\arg\!\min}

\usepackage{jmlr2e}

\ShortHeadings{Direct Approach for Sparse QDA}{Jiang, Wang, and Leng}
\firstpageno{1}

\begin{document}

\title{A Direct Approach for Sparse Quadratic Discriminant Analysis}

\author{\name Binyan Jiang\email by.jiang@polyu.edu.hk \\
       \addr Department of Applied Mathematics\\
      The Hong Kong Polytechnic University\\
       Hong Kong, China
       \AND
              \name Xiangyu Wang \email xiangyuw@google.com \\
       \addr Google LLC\\
       1600 Amphitheatre Pkwy \\
      Mountain view, CA 94043, USA 
       \AND
       \name Chenlei Leng \email C.Leng@warwick.ac.uk \\
       \addr Department of Statistics\\
       University of Warwick\\ and Alan Turing Institute\\
       Coventry, CV4 7AL, UK}

\editor{Saharon Rosset}

\maketitle

\begin{abstract}
Quadratic discriminant analysis (QDA) is a standard tool for classification due to its simplicity and flexibility. Because the number of its parameters scales quadratically with the number of the variables, QDA is not practical, however, when the dimensionality is relatively large.  To address this,  we propose a novel procedure named DA-QDA for QDA in
 analyzing high-dimensional data. Formulated in a simple and coherent framework, DA-QDA aims to directly estimate the key quantities in the Bayes discriminant function including quadratic interactions and a linear index of the variables for classification. Under  appropriate sparsity assumptions, we establish consistency results for estimating the interactions and the linear index, and further demonstrate that the misclassification rate of our procedure converges to the optimal Bayes risk, even when the
 dimensionality is exponentially high with respect to the sample size.   An efficient algorithm based on the alternating direction method
  of multipliers (ADMM) is developed for finding interactions, which  is much faster than its competitor in the literature. The promising performance of DA-QDA is illustrated via extensive simulation studies and the analysis of four real datasets.
\end{abstract}

\begin{keywords}
  Bayes Risk, Consistency, High Dimensional Data, Linear Discriminant Analysis, Quadratic Discriminant Analysis, Sparsity
\end{keywords}

\section{Introduction}

Classification is a central topic in statistical learning and data analysis. Due to its simplicity for producing quadratic decision boundaries, quadratic discriminant analysis (QDA) has become an important technique for classification, adding an extra layer of flexibility to the linear discriminant analysis (LDA); see \cite{Hastie:etal:2009}. Despite its usefulness, the number of the parameters needed by QDA scales squarely with that of the variables, making it quickly inapplicable for problems with large or even moderate dimensionality. This problem is extremely eminent in the era of big data, as one often encounters datasets with the dimensionality larger, often times substantially larger, than the sample size. This paper aims to develop a novel classification approach named DA-QDA, short for \underline{D}irect \underline{A}pproach for \underline{QDA} to make QDA useful for analyzing ultra-high dimensional data.

For ease of presentation, we focus on binary problems where observations are from two classes.
Suppose that the observations from class $1$ follow $X \sim N(\mu_1, \Sigma_1)$ and those from class $2$ satisfy $Y
\sim N(\mu_2, \Sigma_2)$, where $\mu_k \in \mathbb{R}^p, k=1, 2$ are the mean vectors and $\Sigma_k \in \mathbb{R}^{p\times p}, k=1, 2$ are the two covariance matrices.  Compared with LDA,  it is assumed  that $\Sigma_1 \not= \Sigma_2$ in QDA, which gives rise to a class boundary that is quadratic in terms of the variables. Bayes' rule classifies a new observation $z$ to class $1$ if $\pi_1 f(z|\mu_1,\Sigma_1) > \pi_2 f(z|\mu_2, \Sigma_2)$, where $f(z|\mu,\Sigma)$ is the probability density function of a multivariate normal distribution with mean $\mu$ and variance $\Sigma$, and $\pi_1$ and $\pi_2$ are the two prior probabilities. Following simple algebra, the Bayes discriminant
function for a new observation $z$ is seen as
\[ D(z)=(z-\mu)^T \Omega (z-\mu)+\delta^T (z-\mu)+ \eta ,
\]
where $\mu=(\mu_1+\mu_2)/2$ is the mean of the two centroids,
 $\Omega=\Sigma_2^{-1}-\Sigma_1^{-1}$  is the difference of the two precision matrices,
 $\delta =
(\Sigma_1^{-1}+\Sigma_2^{-1})(\mu_1-\mu_2)$, and  $\eta=2\log (\pi_1/\pi_2)+\frac{1}{4}(\mu_1-\mu_2)^T\Omega(\mu_1-\mu_2)+\log |\Sigma_2|-\log |\Sigma_1|$; see for example \cite{Anderson:1958aa}. Note that the discriminant function becomes that of LDA when $\Sigma_1=\Sigma_2=\Sigma$. Completely analogous to a two-way interaction model in linear regression, $\delta$ in  $D(z)$ can be seen as a linear index of the variables whose nonzero entries play the role of main effects, whereas
the nonzero entries in $\Omega$ can be understood as interactions of second-order between the variables. Although there are other ways to represent the discriminant function, $D(z)$ is used as it is a quadratic function of $z-\mu$, making the discriminant function location-invariant with respect to the coordinates.
For easy reference,
we shall call subsequently the parameters $\Omega, \delta, \mu$, and $\eta$ in the Bayes discriminant function collectively as Bayes components.

\subsection{Our contributions}
We highlight the main contributions of this paper as follows.
\begin{itemize}
\item[1.] DA-QDA is the first \textit{direct} approach for sparse QDA in a high dimensional setup. That is, $\Omega, \delta, \mu$, and $\eta$  in the Bayes discriminant function are directly estimated with only sparse assumptions on $\Omega$ and $\delta$ but not on other intermediate quantities;
\item[2.] For estimating $\Omega$, an intermediate step of DA-QDA and a problem of interest in its own right, we develop a new algorithm which is much more computationally and memory efficient than its competitor. See Section 2.1;
\item[3.] We develop new theory to show the theoretical attractiveness of the DA-QDA. In particular,  the theory for estimating $\delta$ is new. See Section 3;
\item[4.] The problem of finding the right intercept $\eta$ is of considerable interest but a general theory on estimated $\eta$ is lacking \citep{Hastie:etal:2009}. We provide a first theory for the convergence property of our estimated $\eta$. See Section 3.4;
\item[5.] In extensive simulation study and real data analysis, the DA-QDA approach outperforms many of its competitors, especially when variables under 
considerable interact. See Section 4.
\end{itemize}

\subsection{Literature review}
As more and more modern datasets are high-dimensional,  the problem of classification in this context has received increasing attention as the usual practice of using empirical estimates for the Bayes components is no longer applicable.
\cite{Bickel2004} first highlighted that LDA is equivalent to random guessing in the worst case scenario when the dimensionality is larger than the sample size. Scientifically and practically in many problems,  however, the components in the Bayes discriminant function can be assumed sparse. In the problem we study in this paper, loosely speaking, the notion of sparsity entertains that the two-way interaction representation of the model only admits a small number of main effects and interactions.
In the past few years, a plethora of methods built on suitable sparsity assumptions have been proposed to estimate the main effects as in LDA; see for example \cite{Shao}, \cite{Cai}, \cite{Fan2012}, \cite{Mai:2012aa}, \cite{Mai:Zou:2013}, and \cite{Jiang:etal:2015}. Other related linear methods for high-dimensional classification can be found in \cite{Leng:2008}, \cite{Witten:Tibshirani:2011}, \cite{Pan:etal:2015}, \cite{Mai:etal:2015}, among others.

As pointed out by \cite{Fan:2015aa} and \cite{Sun:Zhao:2015},  it has been increasingly recognized that the assumption of a common covariance matrix across different classes, needed by LDA, can be restrictive in many practical problems. The extra layer of flexibility offered by QDA that deals with two-way variable interactions makes it extremely attractive for such problems. \cite{Li:2013aa} studied sparse QDA by making sparsity assumptions on $\mu_2-\mu_1, \Sigma_1, \Sigma_2$ and $\Sigma_1-\Sigma_2$ and proposed their sparse estimates. The assumptions made are not directly on the key quantities needed in the discriminant function $D(z)$. In addition, good estimates of these four quantities do not necessarily translate to better classification, a phenomenon similarly argued and observed by \cite{Cai} and \cite{Mai:2012aa} for LDA. \cite{Fan:2015aa} proposed a screening method to identify interactions when $\Omega$ admits a two block sparse structure after permutation, before applying penalized logistic regression on the identified interactions and all the main effects to estimate a sparser model. Their method cannot deal with problems where the support of $\Omega$ is in general positions. 
Further,  the use of a separate second-step penalized logistic regression to determine important interactions and main effects is less appealing from a methodological perspective.  \cite{Fan:etal:2015} suggested a Rayleigh quotient based method for which all the fourth cross-moments of the predictors have to be estimated. Despite all these efforts, a direct yet simple approach for QDA with less stringent assumptions than in \cite{Li:2013aa} for high-dimensional analysis is missing.

The proposed DA-QDA approach in this paper aims to overcome the difficulties mentioned above. In particular, compared with \cite{Li:2013aa}, we only make sparsity assumptions on $\Omega$ and $\delta$ and estimate these two quantities directly in DA-QDA. Compared to \cite{Fan:2015aa}, we allow the interactions in $\Omega$ in general positions, without resorting to a second stage approach for interactions and main effects selection. Compared with \cite{Fan:etal:2015}, we operate directly on QDA for which only second cross-moments of the variables are needed.

DA-QDA can also be understood as a novel attempt to select interactions in the discriminant function that correspond to the nonzero entries in $\Omega$. The problem of interaction selection is a problem of its own importance and has been studied extensively recently for regression problems. See, for example,  \cite{Hao:Zhang:2014} and references therein. 
The problem of estimating $\Omega$ alone has also attracted attention lately in a different context.  To understand how the structure of a  network differs between different conditions and to find the common structures of two different Gaussian graphical models, \cite{Zhao:2014aa} proposed a direct approach for estimating $\Omega$ by formulating their procedure via the Dantzig selector. A severe limitation is that their linear programming procedure needs to deal with $O(p^2)$ constraints, and the memory requirement by the large constraint matrix is of the order $O(p^4)$. As a result, an iteration of the algorithm in \cite{Zhao:2014aa} requires $O(sp^4)$ computations, where $s$ is the cardinality of the support of $\Omega$.
 Apparently, their method does not scale well to high dimensional data. In \cite{Zhao:2014aa}, problems with maximum size $p=120$ were attempted and it was reported that a long time was needed to run their method. In contrast,
we use a lasso formulation and develop a new algorithm based on the alternating direction methods of multipliers (ADMM)
 for estimating $\Omega$. The memory requirement of our algorithm is of the order $O(p^2)$ and its computational cost is of the order $O(p^3)$ per iteration, enabling DA-QDA to easily handle much larger problems.

The rest of the paper is organized as follows. Section 2 outlines the main DA-QDA methodology for estimating $\Omega$ and $\delta$. A novel algorithm based on ADMM for estimating $\Omega$ is developed. Section 3 investigates the theory of DA-QDA and provides various consistency results for estimating $\Omega$, $\delta$, and $\eta$, as well as establishing the consistency of the misclassification risk relative to the Bayes risk. Section 4 presents extensive numerical studies and analysis of four real datasets. Comparison with other classification methods demonstrates that DA-QDA is very competitive in estimating the sparsity pattern and the parameters of interest.
We provide a short discussion and outline future directions of research in Section 5. All the proofs are relegated to the Appendix.

\section{DA-QDA Methodology}

To obtain an estimator for the Bayes discriminant function $D(z)$, we propose direct estimators for the two of its Bayes components $\Omega=\Sigma_2^{-1}-\Sigma_1^{-1}$ and $\delta=(\Sigma_1^{-1}+\Sigma_2^{-1})(\mu_1-\mu_2)$ under appropriate sparsity assumptions.
Given data $X_{j},~j=1,..,n_1$ from class 1 and $Y_{k},~k=1,...,n_2$ from class 2, we can estimate $\mu_i$ and $\Sigma_i,~i=1,2,$ via their sample versions as
\[  \hat\mu_1 = \frac{1}{n_1}\sum_{j=1}^{n_1} X_{j},~\hat\mu_2 = \frac{1}{n_2}\sum_{j=1}^{n_2} Y_{j};\]
\[\hat\Sigma_1= \frac{1}{n_1} \sum_{j=1}^{n_1} (X_{j}-\hat\mu_1)(X_{j}-\hat\mu_1)^T,~\hat\Sigma_2 = \frac{1}{n_2} \sum_{j=1}^{n_2} (Y_{j}-\hat\mu_2)(Y_{j}-\hat\mu_2)^T.
\]
When $p>>\max\{n_1, n_2\}$, $\hat\Sigma_1$ and $\hat\Sigma_2$ are degenerate and cannot be directly used for estimating $\Omega$.
Denote the DA-QDA estimates of $\Omega$ as $\hat\Omega$ and $\delta$ as $\hat{\delta}$ which will be obtained as in (\ref{eq:DA-QDA}) and (\ref{eq:delta}) respectively.
For a given scalar $\eta$, our DA-QDA procedure
 classifies a new observation $z$ in to class $1$ if
 \begin{eqnarray}
 \left[z-\frac{\hat{\mu}_1+\hat{\mu}_2}{2}\right]^T \hat{\Omega}  \left[z-\frac{\hat{\mu}_1+\hat{\mu}_2}{2}\right]+\hat{\delta}^T  \left[z-\frac{\hat{\mu}_1+\hat{\mu}_2}{2}\right]+ \eta>0,
 \label{eq:DA-QDA00}
\end{eqnarray}
and classifies $z$ into class $2$ otherwise. From \eqref{eq:DA-QDA00}, we emphasize again that the nonzero entries in $\hat\Omega$ are the interactions of the variables that contribute to the classification rule, while the nonzero entries in $\hat\delta$ are the main effects of the variables that are used for classification. In the linear discriminant analysis when $\Sigma_1=\Sigma_2$, the rule in \eqref{eq:DA-QDA00} becomes the LDA rule which is linear in the variables. As $\eta=2\log (\pi_1/\pi_2)+\frac{1}{4}(\mu_1-\mu_2)^T\Omega(\mu_1-\mu_2)+\log |\Sigma_2|-\log |\Sigma_1|$ is a scalar, we can choose $\eta$ as $\hat\eta$ using a simple grid search, bypassing the need to estimate the determinants of $\Sigma_1$ and $\Sigma_2$.  This is the strategy implemented in Section 4 and its analytical justification is provided in Section 3.4. Thus in the following, we shall focus on the estimation of $\Omega$ and $\delta$, under certain sparsity assumptions on these two quantities.

 \subsection{Estimating $\Omega$}
Recall $\Omega=\Sigma_2^{-1}-\Sigma_1^{-1}$. It may be attempting to first estimate $\Sigma_1^{-1}$ and $\Sigma_2^{-1}$  as intermediate quantities before taking their difference. It is known, however, that accurate estimation of a covariance matrix or its inverse can be difficult in general in high dimensions unless additional assumptions are imposed (cf. \cite{Bickel:2008aa}). Because $\Omega$ is the quantity of interest that appears in the Bayes' rule, we propose to estimate it directly. To proceed, we note that
 $\Sigma_2\Omega\Sigma_1=\Sigma_1\Omega\Sigma_2=\Sigma_1-\Sigma_2$. If we define a loss function as $Tr\left(\Omega^T\Sigma_1\Omega \Sigma_2\right) /2-Tr\left( \Omega(\Sigma_1-\Sigma_2) \right)$, the loss function is minimized when $\Omega$ satisfies $\Sigma_2\Omega\Sigma_1=\Sigma_1-\Sigma_2$  or $\Omega=\Sigma_2^{-1}-\Sigma_1^{-1}$.
This simple observation motivates the following penalized loss formulation for estimating $\Omega$ by replacing $\Sigma_j, j=1,2$ by their empirical estimates as
\begin{equation}
\hat{\Omega}=\arg\min_{\Omega\in R^{p\times p}} ~\frac{1}{2}Tr\left(\Omega^T\hat\Sigma_1\Omega \hat\Sigma_2\right) -Tr\left( \Omega(\hat\Sigma_1-\hat\Sigma_2) \right)+\lambda\| \Omega \|_1,
\label{eq:DA-QDA}
\end{equation}
where $\|\Omega\|_1$ is the $\ell_1$ penalty of the vectorized $\Omega$ to encourage sparsity and $\lambda$ is the tuning parameter. To obtain a symmetric estimator for $\Omega$, we may simply use $\hat{\Omega}_0=\frac{1}{2}(\hat{\Omega}+\hat{\Omega}^T)$ after $\hat\Omega$ is obtained.
Because the second derivative of the above loss function is $\hat\Sigma_2\otimes\hat\Sigma_1$ which is nonnegative definite, the formulation in \eqref{eq:DA-QDA} is a  convex problem and can be solved by a convex optimization algorithm.

We now develop an ADMM algorithm to solve for $\hat\Omega$ in \eqref{eq:DA-QDA} \citep{Boyd:etal:2011,Zhang:2014aa}.
First write the optimization problem in \eqref{eq:DA-QDA} as 
\begin{equation}
\min_{\Omega\in R^{p\times p}} ~\frac{1}{2}Tr\left(\Omega^T\hat\Sigma_1\Omega \hat\Sigma_2\right) -Tr\left( \Omega(\hat\Sigma_1-\hat\Sigma_2) \right)+\lambda\| \Psi\|_1, ~s.t. ~\Psi=\Omega.
\label{eq:DA-QDA1}
\end{equation}
From this, we can form the augmented Lagrangian as
\begin{eqnarray*}
 L(\Omega,\Psi,\Lambda)&=&
\frac{1}{2}Tr\left(\Omega^T\hat\Sigma_1\Omega \hat\Sigma_2\right) - Tr\left( \Omega(\hat\Sigma_1-\hat\Sigma_2) \right)+\lambda\| {\Psi} \|_1\\
&&+Tr\left(\Lambda (\Omega-{\Psi})\right)+\frac{\rho}{2}\|\Omega-\Psi\|_F^2,
\end{eqnarray*}
where $\|\cdot\|_F$ is the Frobenius norm of a matrix and $\rho$ is a parameter in the ADMM algorithm. See Section 4 for more details.
Given the current estimate $\Omega^k, \Psi^k, \Lambda^k$, we update successively
\begin{eqnarray} 
\Omega^{k+1}=\argmin_{\Omega\in R^{p\times p}} L(\Omega, \Psi^k, \Lambda^k),\label{eq:Lambda} \\
\Psi^{k+1}=\argmin_{\Psi\in R^{p\times p}}L(\Omega^{k+1}, \Psi, \Lambda^k),\label{eq:Omega} \\
\Lambda^{k+1}=\Lambda^k+\rho(\Omega^{k+1}-{\Psi}^{k+1}). \label{eq:Gamma}
\end{eqnarray}
Write $\hat\Sigma_i=U_iD_iU_i^T$ as the eigenvalue decomposition of $\hat\Sigma_i$ where $D_i=\mbox{diag}(d_{i1},\cdots,d_{ip}),~i=1,2$. Denote $A^k=(\hat\Sigma_1-\hat\Sigma_2)-\Lambda^k+\rho{\Psi}^k$ and organize the diagonals of $(D_2\otimes D_1+\rho I)^{-1} $ in a matrix $B$, where $B_{jk}=1/(d_{1j}d_{2k}+\rho)$.
The following proposition provides explicit solutions for \eqref{eq:Lambda} and \eqref{eq:Omega} which ensures efficient updation of our algorithm in each step.
\begin{prop}\label{solution}
Given $\Psi^k$, $\Lambda^k, \rho$ and $\lambda$, the solution for (\ref{eq:Lambda}) is given as:
 \begin{equation*}
 \Omega^{k+1} = U_1 [B \circ (U_1^T A^kU_2)] U_2^T;
 \end{equation*}
Given $\Omega^{k+1}, \Lambda^k$ and $\rho$, the solution for (\ref{eq:Omega}) is given as:
 \begin{equation} \Psi^{k+1}= S(\Omega^{k+1}+\frac{\Lambda^k}{\rho}, \frac{\lambda}{\rho}),\label{eq:UpdateGamma}
  \end{equation}
where $S$ is known as the soft-thresholding operator on a matrix. Namely,
  the $(i,j)$ entry of $S(A, b)$ for a matrix $A=(a_{ij})$ is $\mbox{sign}(a_{ij})(|a_{ij}|-b)_+$ where $(c)_+=c$ for $c>0$ and $(c)_+=0$
 otherwise.
\end{prop}
Note that for a given $\rho$, when updating $\Omega$, we only need to update $A^k$ which involves simple matrix subtraction, and then use matrix multiplication. Therefore the update in \eqref{eq:Lambda} can be efficiently implemented.
Following is a brief derivation on how we obtain the explicit solutions given in Proposition \ref{solution}.
For \eqref{eq:Lambda}, note that the derivative of $L$  with respect to $\Omega$ is
\[ \hat\Sigma_1\Omega\hat\Sigma_2-(\hat\Sigma_1-\hat\Sigma_2)+\Lambda^k+\rho(\Omega-\Psi^k)=(\hat\Sigma_1\Omega\hat\Sigma_2+\rho\Omega)-(\hat\Sigma_1-\hat\Sigma_2)+\Lambda^k-\rho\Psi^k,
\] which can be written as
\[ (\hat\Sigma_2 \otimes \hat\Sigma_1+\rho I) vec(\Omega)=vec\left((\hat\Sigma_1-\hat\Sigma_2)-\Lambda^k+\rho\Psi^k\right),
\]
where $vec$ is the vector operator. We have
\[
 vec(\Omega)=(\hat\Sigma_2 \otimes \hat\Sigma_1+\rho I) ^{-1}vec\left((\hat\Sigma_1-\hat\Sigma_2)-\Lambda^k+\rho\Psi^k\right).
 \]
Using the equality $vec(AXB)=(B^T\otimes A)vec(X)$, and
 \begin{eqnarray*} 
 (\hat\Sigma_2 \otimes \hat\Sigma_1+\rho I) ^{-1}&=&[(U_2\otimes U_1)(D_2\otimes D_1+\rho I)(U_2^T\otimes U_1^T)]^{-1}\\
 &=&(U_2\otimes U_1)(D_2\otimes D_1+\rho I)^{-1}(U_2^T\otimes U_1^T),
 \end{eqnarray*}
we have
 \begin{eqnarray*}
 &&(\hat\Sigma_2 \otimes \hat\Sigma_1+\rho I) ^{-1} vec(A^k)\\
 &=&[(U_2\otimes U_1)(D_2\otimes D_1+\rho I)(U_2^T\otimes U_1^T)]^{-1}vec(A^k)\\
&=&(U_2\otimes U_1)(D_2\otimes D_1+\rho I)^{-1} (U_2^T \otimes U_1^T)vec(A^k)\\
&=&(U_2\otimes U_1)(D_2\otimes D_1+\rho I)^{-1} vec(U_1^T A^kU_2)\\&=&(U_2\otimes U_1) vec\left(B \circ (U_1^T A^kU_2)\right)\\
&=&vec\left(U_1 [B \circ (U_1^T A^kU_2)] U_2^T\right)
 \end{eqnarray*}
 where $\circ$ is the Hadamard product.
 Therefore,
 \begin{equation*}
 \Omega = U_1 [B \circ (U_1^T A^kU_2)] U_2^T.
 \end{equation*}
 Next we examine \eqref{eq:Omega}. Ignoring terms that are independent of $\Psi$, we just need to minimize
 \[ \frac{\rho}{2}Tr({\Psi}^T{\Psi})-\rho Tr((\Omega^{k+1})^T{\Psi})-Tr((\Lambda^k)^T{\Psi})+\lambda\| {\Psi}\|_1,
  \]
  and the solution can be easily seen as \eqref{eq:UpdateGamma}. Again, the update for $\Gamma$ can be efficiently implemented.

 Our algorithm can be now summarized as following.
 \begin{itemize}
 \item[1.] Initialize $\Omega$, ${\Psi}$ and $\Lambda$. Fix $\rho$. Compute SVD $\hat\Sigma_1=U_1D_1U_1^T$ and $\hat\Sigma_2=U_2D_2U_2^T$, and compute $B$ where $B_{jk}=1/(d_{1j}d_{2k}+\rho)$.
  Repeat steps 2-4 until convergence;
 \item[2.] Compute $A =(\hat\Sigma_1-\hat\Sigma_2)-\Lambda+\rho{\Psi}$ . Then update $\Omega$ as $\Omega = U_1 [B \circ (U_1^T AU_2)] U_2^T$;
 \item[3.] Update ${\Psi}$ by soft-thresholding $\Omega+\frac{\Lambda}{\rho}$ elementwise by $\frac{\lambda}{\rho}$;
 \item[4.] Update $\Lambda$ by $\Lambda\leftarrow \Lambda+\rho(\Omega-{\Psi})$.
 \end{itemize}
 Note that the algorithm involves singular value decomposition of $\hat\Sigma_1$ and $\hat\Sigma_2$ only once. The rest of the algorithm only involves matrix addition and multiplication. Thus, the algorithm is extremely efficient. Compared with \cite{Zhao:2014aa} whose  algorithm has computational complexity of the order at least $O(p^4)$ and a memory requirement of $O(p^4)$,  our algorithm has a memory requirement of the order $O(p^2)$ and computational complexity of $O(p^3)$. As a result, our method can handle much larger problems. 
 As a first order method for convex problems, the convergence of ADMM algorithms is in general of rate $O(k^{-1})$, where $k$ is the number of iterations. Convergence analysis of ADMM algorithms under different assumptions has been well established in some very recent optimization literatures; see for example, \cite{Nishihara2015}, \cite{Hong2017} and \cite{Chen2017}. By verifying the assumptions in \cite{Hong2017}, we can established similar linear convergence results for our algorithm; see Lemma \ref{linearADMM} in the Appendix for more details.

 \subsection{The linear index $\delta$}
 After having estimated $\Omega$ as $\hat{\Omega}$,  we discuss the estimation of the linear index $\delta=(\Sigma_1^{-1}+\Sigma_2^{-1})(\mu_1-\mu_2)$.  We develop a procedure that avoids estimating $\Sigma_1^{-1}$ and $\Sigma_2^{-1}$.
 First note that
\[\Sigma_1\delta=\Sigma_1(\Sigma_1^{-1}+\Sigma_2^{-1})(\mu_1-\mu_2) =2(\mu_1-\mu_2)+\Sigma_1\Omega(\mu_1-\mu_2),\]
\[\Sigma_2\delta=\Sigma_2(\Sigma_1^{-1}+\Sigma_2^{-1})(\mu_1-\mu_2) =2(\mu_1-\mu_2)-\Sigma_2\Omega(\mu_1-\mu_2),\]
and
\[ (\Sigma_1+\Sigma_2)\delta=4(\mu_1-\mu_2)+(\Sigma_1-\Sigma_2)\Omega(\mu_1-\mu_2).\]
The last equation is the derivative of $\delta^T (\Sigma_1 +\Sigma_2)\delta/2-\{4(\mu_1-\mu_2)+(\Sigma_1-\Sigma_2)\Omega(\mu_1-\mu_2)\}^T \delta$. Motivated by this,
we estimate $\delta$ by a direct method using the lasso regularization, similar to the one in \cite{Mai:2012aa}, as
\begin{equation}
\hat\delta=
\arg\min_{\delta\in R^p} ~\frac{1}{2} \delta^T (\hat\Sigma_1 +\hat\Sigma_2)\delta-\hat{\gamma}^T \delta +\lambda_\delta \| \delta\|_1,\label{eq:delta}
\end{equation}
where $\hat{\gamma}=4(\hat\mu_1-\hat\mu_2)+(\hat\Sigma_1-\hat\Sigma_2)\hat\Omega(\hat\mu_1-\hat\mu_2)$, $\|\cdot\|_1$ is the vector $\ell_1$ penalty and $\lambda_\delta$ is a tuning parameter. The optimization in \eqref{eq:delta} is a standard lasso problem and is easy to solve using existing lasso algorithms. We remark that \eqref{eq:delta} is much more challenging to analyze theoretically than the method in \cite{Mai:2012aa}, since the accuracy of $\hat\Omega$ as an estimator of $\Omega$ has to be carefully quantified in $\hat\gamma$. 

We emphasize that our framework is extremely flexible and can accommodate additional constraints. As a concrete example, let's consider enforcing the so-called strong heredity principle in that an interaction is present unless the corresponding main effects are both present, i.e.  if $\Omega_{jk}\not=0$ then $\delta_j\not=0$ and $\delta_k\not=0$;
 see for example \cite{Hao:Zhang:2015}. Denote $\mathcal{I} \subset \{1, \ldots, p\}$ as the set such that for any $j, k\in \mathcal{I}$ there exists some $\hat\Omega_{jk}\not=0$. We can change the penalty in \eqref{eq:delta} as $\|\delta_{\mathcal{I}^C}\|_1$ such that the variables in $\mathcal{I}$ are not penalized. Due to space limitation, this line of research will not be studied in the current paper.

\section{Theory}
We show that our method can consistently select the true nonzero interaction terms in $\Omega$ and the true nonzero terms in $\delta$. In addition, we provide explicit upper bounds for the estimation error under $l_\infty$ norm. For classification, we further show that the misclassification rate of our DA-QDA rule converges to the optimal Bayes risk under some sparsity assumptions. For simplicity in this section we assume that $n_1 \asymp n_2$ and write $n=\min\{n_1,n_2\}-1$. 
Instead of assuming $\mu_2-\mu_1$, $\Sigma_1,\Sigma_2$ and $\Sigma_1-\Sigma_2$ to be sparse as in \cite{Li:2013aa}, we only assume that $\Omega$ and $\delta$ are sparse.
For the estimation of $\Omega$, the rate in Corollary \ref{cor1} is similar to the one in Theorem 3 of \cite{Zhao:2014aa}. However, as we pointed out previously, our method is computationally much more efficient and scales better to large-dimensional problems. In addition, our work is the first direct estimation approach for sparse QDA. Importantly, the results for estimating $\delta$ are new.

Note that when estimating $\delta$ as in (\ref{eq:delta}), we have used $\hat{\Omega}$ as a plug-in estimator for $\Omega$. Consequently, from Corollaries \ref{cor1} and \ref{cor2}, the error rate of $\hat{\delta}$ in estimating $\delta $ is a factor times of that of $\hat{\Omega}$ in estimating $\Omega$. However, in the DA-QDA discriminant function defined as in (\ref{eq:DA-QDA00}), $\hat{\Omega}$ appears in the first term which is a product of three components while $\hat{\delta}$ appears in the second term which is a product of two components. As a consequence, the overall estimation error rates of these two terms become equal. This implies that even though the estimating error of $\hat{\Omega}$ might aggregate in the estimation of $\delta$, it does not affect the convergence rate of the classification error at all. Below we provide theory for estimating $\Omega$, $\delta$, and $\eta$, as well as quantifying the overall misclassification error rate.

\subsection{Theory for estimating $\Omega$}
We first introduce some notation.
We assume that $\Omega=(\Omega_{ij})_{1\leq i,j\leq p}$ is sparse with support $S=\{(i,j): \Omega_{ij}\neq 0\}$ and we use $S^c$ to denote the complement of $S$. Let $d$ be the maximum node degree in $\Omega$. For a vector $x=(x_1,\ldots,x_p)^T$, the $l_q$ norm is defined as $|x|_q=(\sum_{i=1}^p|x_i|^q)^{1/q}$ for any $1\leq q<\infty$ and the $l_\infty$ norm is defined as $|x|_\infty=\max_{1\leq i\leq p}|x_i|$. For any matrix $M=(m_{ij})_{1\leq i,j\leq p}$, its entrywise $l_1$ norm is defined as $||M||_1=\sum_{1\leq i,j\leq p}|m_{ij}|$  and its entrywise $l_\infty$ norm is written as $||M||_\infty=\max_{1\leq i,j\leq p}|m_{ij}|$. We use $||M||_{1,\infty}=\max_{1\leq i\leq p }{\sum_{j=1}^p|m_{ij}|}$ to denote the $l_1/l_\infty$ norm induced matrix-operator norm. We denote $\Gamma=\Sigma_2\otimes \Sigma_1$ and $\hat{\Gamma}=\hat{\Sigma}_2\otimes\hat{\Sigma}_1$. Write $\Sigma_k=(\sigma_{kij})_{1\leq i,j\leq p}$, $\hat{\Sigma}_k=(\hat{\sigma}_{kij})_{1\leq i,j\leq p}$ for $k=1,2$. By the definition of Kronecker product, $\Gamma$ is a $p^2\times p^2$ matrix indexed by vertex pairs in that $\Gamma_{(i,j),(k,l)}=\sigma_{1ik}\sigma_{2jl}$.
Denote $\Delta_i=\hat{\Sigma}_i-\Sigma_i$ for $i=1,2$, $\Delta_\Gamma=\hat{\Gamma}-\Gamma, \Delta_{\Gamma^T}=\hat{\Gamma}^T-\Gamma^T$, and $\epsilon_i=||\Delta_i||_{\infty}$, $\epsilon=\max\{\epsilon_1,\epsilon_2\}$. $B=\max\{||\Sigma_1||_{\infty},||\Sigma_2||_{\infty}\} $, $B_\Sigma=\max\{||\Sigma_1||_{1,\infty},||\Sigma_2||_{1,\infty}\} $ and $B_\Gamma=||\Gamma_{S,S}^{-1}||_{1,\infty}$, $B_{\Gamma^T}=||(\Gamma_{S,S}^{T})^{-1}||_{1,\infty}$, $B_{\Gamma,\Gamma^T}=\max\{B_\Gamma,B_{\Gamma^T}\}$.

To establish the model selection consistency of our estimator, we assume the following irrepresentability condition:
\[
\alpha=1-\max_{e\in S^c}|\Gamma_{e,S}\Gamma_{S,S}^{-1}|_1>0.
\]
This condition was first introduced by \cite{Zhao:Yu:2006} and \cite{Zou2006} to establish the model selection consistency of the lasso.
The following theorem gives the model selection consistency  and the rate of convergence for the estimation of $\Omega$.

\begin{thm}\label{thmomega}
Assume that $\alpha>0$ and $d^2 B^2B_\Sigma^2B_{\Gamma,\Gamma^T}^2   
 \sqrt{\frac{\log p }{ n}}
\rightarrow 0$. For any $c>2$, by choosing 
$
 \lambda= \kappa_1 d^2 B^2B_\Sigma^2B_{\Gamma,\Gamma^T}^2   
 \sqrt{\frac{\log p }{ n}}
$ for some large enough constant $\kappa_1>0$, we have with probability greater than $1-p^{2-c}$,
\begin{itemize}
\item[(i)] $\hat{\Omega}_{S^c}=0$;
\item[(ii)] there exists a large enough constant $\kappa_2>0$ such that
\begin{eqnarray*}
||\hat{\Omega}-\Omega||_\infty< \kappa_2 d^2 B^2B_\Sigma^2B_{\Gamma,\Gamma^T}^2  \sqrt{\frac{\log p  }{ n}}.
\end{eqnarray*}
\end{itemize}
 \end{thm}
  
Theorem \ref{thmomega} states that if the irrepresentability condition is satisfied, the support of $\Omega$ is estimated consistently, and the rate of convergence of estimating $\Omega$ under $l_\infty$ norm is of order $O\left(d^2 B^2B_\Sigma^2B_{\Gamma,\Gamma^T}^2\sqrt{\frac{\log p }{ n}}\right)$, which depends on the sparsity of $\Sigma_1, \Sigma_2$ and their Kronecker product. For example, our assumption $d^2 B^2B_\Sigma^2B_{\Gamma,\Gamma^T}^2  \sqrt{\frac{\log p  }{ n}}\rightarrow 0$ implies that $B=\max\{||\Sigma_1||_{\infty},||\Sigma_2||_{\infty}\}$ can diverge in a rate of $o(d^{-1} B_\Sigma^{-1}B_{\Gamma,\Gamma^T}^{-1}  ({\frac{ n}{\log p  }})^{1/4})$.
From the proof of Theorem \ref{thmomega}, we have the following corollary.
\begin{cor}\label{cor1}
Assume that $\alpha>0$,  
$B_\Sigma<\infty$ and $B_{\Gamma,\Gamma^{T}}<\infty$. For any constant $c>2$, choosing $\lambda=Cd^2\sqrt{\frac{\log p}{n}}$ for some constant $C>0$, if $d^2\sqrt{\frac{\log p}{n}}\rightarrow 0$, we have with probability greater than $1-p^{2-c}$, $\hat{\Omega}_{S^c}=0$ and
$$||\hat{\Omega}-\Omega||_\infty=O\left(d^2\sqrt{\frac{\log p}{n}}\right).$$
\end{cor}
Similar to Condition 2 in \cite{Zhao:2014aa}, the assumption  $B_{\Gamma,\Gamma^T}<\infty$ in Corollary \ref{cor1} is closely related to the mutual incoherence property introduced in \cite{Donoho}. In fact, it holds when imposing the usual mutual incoherence condition on the inverses of the submatrices (indexed by $S$) of $\Sigma_1$ and $\Sigma_2$.
Since $d$ is the maximum node degree in $\Omega$, the number of nonzero entries in $\Omega$ is of order $O(dp)$. 
the rate $O\left(d^2\sqrt{\frac{\log p}{n}}\right)$ we obtained in Corollary \ref{cor1} is better than the rate in Theorem 3 of \cite{Zhao:2014aa}. However, in the case where only O(d) covariates and some of their interactions are important, our rate is the same as the one in \cite{Zhao:2014aa}.

\subsection{Theory for estimating $\delta$}

 Let $D=\{i: \delta_i\neq 0\}$ be the support of $\delta$ and let $d_\delta$ be its cardinality. Denote
$A_1=||\Omega||_{1,\infty}$, $A_2=||(\Omega^{-1})_{\cdot,D}||_{1,\infty}$,
 $\epsilon_\mu=\max\{|\mu_1-\hat{\mu}_1|_\infty ,|\mu_2-\hat{\mu}_2|_\infty\}$. We define  $A_\Sigma=||\Sigma_{D,D}^{-1}||_{1,\infty}$ where ${\Sigma}=({\Sigma}_1+{\Sigma}_2)/2$  and write $\hat{\Sigma}=(\hat{\Sigma}_1+\hat{\Sigma}_2)/2$, $\gamma=4(\mu_1-\mu_2)+(\Sigma_1-\Sigma_2)\Omega(\mu_1-\mu_2)$, $\Delta_{\mu}={\mu}_1-\mu_2$, $\hat{\Delta}_\mu=\hat{\mu}_1-\hat{\mu}_2$, $A_\gamma=||\gamma||_\infty$ and $||\hat{\Omega}-\Omega||_\infty=\epsilon_\Omega$.

To establish the model selection consistency of our estimator $\hat{\delta}$, we assume the following irrepresentability condition:
\[
\alpha_\delta=1-\max_{e\in D^c}|\Sigma_{e,D}\Sigma_{D,D}^{-1}|_1>0.
\]
Let $d_0=\max\{d,d_\delta\}$.  The following theorem gives the model selection consistency  and the rate of convergence for the estimation of $\delta$.

\begin{thm}\label{thmdelta}
Assume that $\alpha_\delta>0$, $|\Omega(\mu_1-\mu_2)|_1=O(d_0^2)$, $\|\Omega\|_\infty<\infty$, $A_\gamma<\infty$ and $     
 d_0^3A_\Sigma^2    B^2B_\Sigma^3B_{\Gamma,\Gamma^T}^2  \sqrt{\frac{\log p  }{ n}}\rightarrow 0$ . Under the same assumptions in Theorem \ref{thmomega}, for any $c>2$,  by choosing
$
 \lambda_\delta=\kappa_3   
 d_0^3A_\Sigma^2    B^2B_\Sigma^3B_{\Gamma,\Gamma^T}^2     \sqrt{\frac{\log p  }{ n}},
$
for some large enough constant $\kappa_3$,
we have, with probability greater than $1-p^{2-c}$,
\begin{itemize}
\item[(i)] $\hat{\delta}_{D^c}=0$;

\item[(ii)] there exists a large enough constant $\kappa_4>0$ such that,
$$
||\hat{\delta}-\delta||_\infty< \kappa_{4 }  d_0^3A_\Sigma^2    B^2B_\Sigma^3B_{\Gamma,\Gamma^T}^2  \sqrt{\frac{\log p  }{ n}}.
$$

\end{itemize}
 \end{thm}

From Theorem \ref{thmdelta} and Corollary \ref{cor1} we immediately have:

\begin{cor}\label{cor2}
Suppose the assumptions of Corollary \ref{cor1} and Theorem \ref{thmdelta} hold and assume that  
$A_\Sigma<\infty$. For any constant $c>2$,  by choosing $\lambda_\delta=Cd_0^3\sqrt{\frac{\log p}{n}}$ for some large enough constant $C>0$, we have with probability greater than $1-p^{2-c}$,
$$||\hat{\delta}-\delta||_\infty=O\left(d_0^3\sqrt{\frac{\log p}{n}}\right).$$
\end{cor}
 When $\Sigma_1=\Sigma_2$, $\delta$ reduces to $2\Sigma_1^{-1}(\mu_1-\mu_2)$, which is proportional to the direct discriminant variable $\beta$ in \cite{Mai:2012aa}, \cite{Cai} and \cite{Fan2012}, and variables in $D=\{i: \delta_i\neq 0\}$ are linear discriminative features contributing to the Bayes rule.
 From the proof of Theorem \ref{thmdelta} and the rate given in Theorem \ref{thmdelta}, we can see that when $A_\gamma<\infty$, $||\hat{\delta}-\delta||_\infty$ is of order $O(A_\Sigma\lambda_\delta)$. This is consistent to the result obtained in Theorem 1 of \cite{Mai:2012aa} for the $\Sigma_1=\Sigma_2$ case.

\subsection{Misclassification rate}
In this subsection, we study the asymptotic behavior of the misclassification rate for a given $\eta$ and postpone the theory when $\eta$ is estimated to Section 3.4.
Let $R(i|j)$ and $R_n(i|j)$ be the probabilities that a new observation from class $j$ is misclassified to class $i$ by Bayes' rule and the DA-QDA rule respectively. Suppose $2\log (\pi_1/\pi_2)=\eta-\frac{1}{4}(\mu_1-\mu_2)^T\Omega(\mu_1-\mu_2)-\log |\Sigma_2|+\log |\Sigma_1|$.
The optimal Bayes risk is given as
\begin{eqnarray*}
R=\pi_1R(2|1)+\pi_2R(1|2),
\end{eqnarray*}
and the misclassification rate of the DA-QDA rule takes the following form:
\begin{eqnarray*}
R_n=\pi_1 R_n(2|1)+\pi_2 R_n(1|2).
\end{eqnarray*}

Suppose $z_i\sim N(\mu_i,\Sigma_i)$ for $i=1,2$. Denote the density of $(z_i-\mu)^T\Omega(z_i-\mu)+\delta^T(z_i-\mu)+\eta$ as $F_i(z)$. For any constant $c$, define
\begin{eqnarray*}
u_c=\max\{{\esssup}_{z\in [-c, c]}F_i(z), i=1,2\},
\end{eqnarray*}
where $\esssup$ denotes the essential supremum which is defined as supremum on almost everywhere of the support, i.e., except on a set of measure zero.
Let $s:=\|S\|_0$ be the number of nonzero elements in $\Omega$. The following theorem establishes upper bounds for the misclassification rate difference between $R_n$ and $R$.

\begin{thm}\label{thm3}
Assuming that there exist constants $C_{\mu}>0, C_{\Sigma}>1$ such that $\max\{|\mu_1|_\infty,|\mu_2|_\infty\}$ $\leq C_\mu$, and $C_\Sigma^{-1}\leq \min\{\lambda_p(\Sigma_1),\lambda_p(\Sigma_2)\}\leq \max\{\lambda_1(\Sigma_1),\lambda_1(\Sigma_2)\}\leq C_\Sigma$ where $\lambda_i(\Sigma_j)$ denotes the $i$th largest eigenvalue of $\Sigma_j$.
Under the assumptions of Theorems \ref{thmomega} and \ref{thmdelta}, we have:

(i) if $s d_0^2 B^2B_\Sigma^2B_{\Gamma,\Gamma^T}^2   
\sqrt{\frac{\log p }{ n}}+  d_0^2A_\Sigma^2    B^2B_\Sigma^3B_{\Gamma,\Gamma^T}^2  \sqrt{\frac{\log p  }{ n}}\rightarrow 0$ and there exist positive constants $c, U_{c}$ such that $u_c\leq U_{c}<\infty$, then
$$R_n-R=O_p\left( s d_0^2 B^2B_\Sigma^2B_{\Gamma,\Gamma^T}^2   
 \sqrt{\frac{\log p }{ n}}+  d_0^2A_\Sigma^2    B^2B_\Sigma^3B_{\Gamma,\Gamma^T}^2  \sqrt{\frac{\log p  }{ n}} \right);$$

(ii) if $ \left(s d_0^2 B^2B_\Sigma^2B_{\Gamma,\Gamma^T}^2   
 \sqrt{\frac{\log p }{ n}}+  d_0^2A_\Sigma^2    B^2B_\Sigma^3B_{\Gamma,\Gamma^T}^2  \sqrt{\frac{\log p  }{ n}}\right)(1+u_c)\rightarrow 0$ for some positive constant $c$, then with probability greater than $1-3p^{2-c}$ for some constant $c>2$,
\begin{eqnarray*}
R_n-R=O\Bigg((1+u_c)\times \Bigg(s d_0^2 B^2B_\Sigma^2B_{\Gamma,\Gamma^T}^2 \log p  
 \sqrt{\frac{\log p }{ n}}+  d_0^2A_\Sigma^2    B^2B_\Sigma^3B_{\Gamma,\Gamma^T}^2    \frac{\log p  }{ \sqrt{n}}\Bigg)
   \Bigg).
\end{eqnarray*}
\end{thm}

Theorem \ref{thm3} (i) indicates that under appropriate sparsity assumptions, our DA-QDA rule is optimal in that its misclassification rate converges to the optimal Bayes risk in probability. The second statement of Theorem \ref{thm3} states that under stronger conditions, $R_n$ converges to $R$ with overwhelming probability.
From Corollary \ref{cor1} and Corollary \ref{cor2} and the above theorem, we immediately have:
 \begin{cor}\label{cor3}
 Under the assumptions of Corollary \ref{cor1}, Corollary \ref{cor2} and Theorem \ref{thm3}, we have,

 (i) if $sd_0^2\sqrt{\frac{\log p}{ n}}\rightarrow 0$ and there exist positive constants $c, U_{c}$ such that $u_c\leq U_{c}<\infty$, then
$$R_n-R=O_p\left( sd_0^2\sqrt{\frac{\log p}{ n}}\right);$$

(ii) if $(1+u_c)sd_0^2\sqrt{\frac{\log^3 p}{ n}}\rightarrow 0$ for some constant $c>0$, then with probability greater than $1-3p^{2-c}$ for some constant $c>2$,
$$R_n-R=O\left( sd_0^2(1+u_c)\sqrt{\frac{\log ^3p}{ n}}\right).$$
 \end{cor}

We remark that the assumption $u_c\leq U_c$ for some constants $c$ and $U_c$ is similar to Condition (C4) in \cite{Li:2013aa}, and our assumption is weaker in that we only assume the densities $F_i(z)$ is bounded in a neighborhood of zero while Condition (C4) in \cite{Li:2013aa} states that the densities are bounded everywhere. 

\subsection{Choice of $\eta$}
The question of choosing the scalar $\eta$ is critical for classification but receives little attention
in existing literature; see \cite{Mai:2012aa} for a detailed discussion for the LDA case. In this section, we propose to choose $\eta$ by minimizing the in-sample misclassification error and establish analytical results for the estimation of $\eta$ and the misclassification rate. With some abuse of notation, let $(z_i, l_i)$ be our data where $z_i's$ are the covariates and $l_i's$ are the labels, i.e., $l_i\in \{0, 1\}$. To obtain $\hat\eta$, we seek to minimize the in-sample misclassification error given $\hat \mu, \hat\delta$ and $\hat\Omega$:
\begin{align*}
\hat\eta = arg\min_e \frac{1}{n}\sum_{i = 1}^n |I\{\hat D(z_i, e) > 0\} - l_i|,
\end{align*}
where $\hat D(z, e) = (z-\hat\mu)^T\hat\Omega(z-\hat\mu)+\hat\delta^T(z-\hat\mu) + e$. We write $d(z)=(z-\mu)^T\Omega(z-\mu)+\delta^T(z-\mu)$ and $\hat{d}(z)=(z-\hat\mu)^T\hat\Omega(z-\hat\mu)+\hat\delta^T(z-\hat\mu)$ and hence we have
\begin{align*}
\hat D(z_i, e) = \hat d(z_i) + e.
\end{align*}
Then the object function becomes
\begin{align*}
\frac{1}{n}\sum_{i = 1}^n |I\{\hat d(z_i) + e > 0\} - l_i|.
\end{align*}
Without loss of generality, we can assume $\hat  d(z_1) < \hat  d(z_2) <\cdots < \hat  d(z_n)$. For any $e$ we define the index $k(e)$ to be the largest $\hat  d(z_k)$ that satisfies $\hat  d(z_k) < -e < \hat  d(z_{k+1})$. Thus, the optimization can be further simplified as
\begin{align*}
\hat\eta = arg\min_e \frac{1}{n}\bigg[\sum_{i = 1}^{k(e)} l_i + \sum_{i = k(e) + 1}^n (1 - l_i)\bigg].
\end{align*}
Solving the above problem is simple. One just needs to compute the values of the object function for $k = 0, 1, 2, \cdots, n$ and find the index $k^*$ that minimizes its value. The optimal $\hat\eta$ can then be found as any value satisfying
\begin{align*}
\hat\eta \in (-\hat d(z_{k^* + 1}), ~-\hat d(z_{k^*})).
\end{align*}
Next we establish the asymptotic results for $\hat{\eta}$ and the misclassification rate. For a given $e$, we use $R(d,e)$ and $R(\hat{d},e)$ to denote the misclassification rate associated with discriminant function $D(z,e)=d(z)+e$ and discriminant function $\hat{D}(z,e)=\hat{d}(z)+e$ respectively. Analogously, the in-sample misclassification rate of $D(z,e)$ and $\hat{D}(z,e)$ are denoted as $R_n(d,e)$ and $R_n(\hat{d},e)$. From the optimality of the Bayes rule, we know that $\eta=2\log (\pi_1/\pi_2)+\frac{1}{4}(\mu_1-\mu_2)^T\Omega(\mu_1-\mu_2)+\log |\Sigma_2|-\log |\Sigma_1|$ is the unique minimizer of $R(d,\cdot)$ and we denote the corresponding optimal Bayes misclassification rate as $R=R(d,\eta)$. On the other hand, $\hat{\eta}$ is a minimizer of $R_n(\hat{d},e)$. In order to make the estimation problem feasible, we assume that there is exists a constant $c_\eta$ such that $|\eta|\leq c_\eta<\infty$. The following proposition indicates that, although the 0-1 loss used for computing the misclassification rate is neither continuous nor convex, the misclassification rate has a desirable property.

\begin{prop}\label{etalemma1}
$R(d,e)$ is strictly monotone increasing in $e\in[\eta,\infty)$ and strictly monotone decreasing in $e\in(-\infty,\eta]$.
\end{prop}
From Proposition \ref{etalemma1} and following Theorem 5.7 of \cite{Vardervart}, we establish the following theorem, which indicates that the estimator $\hat{\eta}$ is consistent and the resulting misclassification rate using the estimated rule $\hat{D}(z,\hat{\eta})$ tends to the optimal Bayes misclassification rate in probability.
\begin{thm}\label{thm4}
Let $\hat{\eta}$ be a minimizer of $R_n(\hat{d},e)$. Under the assumptions of Theorem \ref{thm3}, we have:

(i) $\hat{\eta}\rightarrow \eta$  in probability;

(ii) $R_n(\hat{d},\hat{\eta})\rightarrow R$ in probability.
\end{thm}

\section{Numerical Study}
In this section, we provide extensive numerical evidence to show the empirical performance of {DA-QDA} by comparing it to its competitors, including the sparse QDA ({sQDA}, \cite{Li:2013aa}), the innovated interaction screening for sparse quadratic discriminant analysis ({IIS-SQDA}, \cite{Fan:2015aa}), penalized logistic regression with only main effects considered ({PLR}), penalized logistic regression with all interaction terms ({PLR2}), the direct approach for sparse LDA ({DSDA}, \cite{Mai:2012aa}), the conventional LDA ({LDA}), the conventional QDA ({QDA}) and the oracle procedure (Oracle). The oracle procedure uses the true underlying model and serves as the optimal risk bound for comparison. We evaluate all methods via nine synthetic datasets and four real datasets. In addition, we also include L1-regularized SVM (L-SVM) and kernel SVM (K-SVM, with Gaussian kernel) performance as a benchmark in the real data analysis.

To fit {DA-QDA}, we employ ADMM to estimate $\Omega$ and the coordinate-wise descent algorithm \citep{Friedman:etal:2010} to fit $\delta$. Once $\Omega$ and $\delta$ are given, we then find the value of $\eta$ by a simple linear search, minimizing the in-sample misclassification error. The rate parameter $\rho$ in ADMM is set according to the optimal criterion suggested by \cite{Ghadimi:etal:2015}. 
The other two tuning parameters, $\lambda$ for estimating $\Omega$ and $\lambda_\delta$ for estimating $\delta$, are chosen by 5-fold cross-validation, where the loss function is chosen to be the out-of-sample misclassification rate.
To reduce the searching complexity, we are currently using a searching path rather than grid based tuning to avoid redundant computation and in-memory parallelism to distribute the computational tasks. It is worth noting that the calculation for each individual tuning pair ($\lambda$, $\lambda_\delta$) can be made completely independent such that it is possible to distribute the entire calculation to multiple threads in a parallel fashion. 
One can also tune the parameters using the objective functions in \eqref{eq:DA-QDA} and \eqref{eq:delta} separately. However, we found that this strategy did not often lead to better classification results than tuning them jointly. This is possibly due to the complex shape of the misclassification surface as a function of these two tuning parameters. We implemented {sQDA} in Matlab with the leave-one-out-cross-validation \citep{Li:2013aa} to tune the three parameters. We employ Matlab's built-in function {\em fitcdiscr} to fit {LDA} and {QDA} and the R package {\em dsda} \citep{Mai:2012aa} to fit {DSDA}. For {PLR}, {PLR2} and the second stage fit of {IIS-SQDA} which is a penalized logistic regression, we use the {\em glmnet} package and set $\alpha = 0.5$ as the elastic net parameter. Other values of $\alpha$ was tried but did not change the result much.
The other tuning parameter in {\em glmnet} is chosen by 10-fold cross-validation to minimize out-of-sample classification error. For the first stage of {\em IIS-SQDA} which is a screening step, we adopt the oracle-assisted approach proposed in \cite{Fan:2015aa}, i.e., using the true $\Sigma_1$ and $\Sigma_2$ to compute the transformed variables used for screening as discussed in \cite{Fan:2015aa}. To seek an appropriate screening size, we preserve the top $10, 30$ or $50$ variables for each experiment to form interaction terms and report the best result (smallest misclassification error) for {\em IIS-SQDA}.

\subsection{Synthetic data}
For synthetic data, we use the same setup in \cite{Fan:2015aa}. Observations are simulated from $N(u_1, \Sigma_1)$ and $N(u_2, \Sigma_2)$ where $u_2 = 0$. Recall $\Omega_1 = \Sigma_1^{-1}$ and $\Omega_2 = \Sigma_2^{-1}$. We set $u_1 = \Sigma_1\beta$ for $\beta = (0.6, 0.8, 0, \cdots, 0)^T$ . We consider three different dimensions $p = 50, ~200, $ or $500$ with $n_1 = n_2 = 100$. The parameters $\Omega_1, \Omega_2$ and $\beta$ are set as follows.
\begin{itemize}
\item { Model 1:} This model is Model 3 in \cite{Fan:2015aa} where $\Omega_1$ is a band matrix with $(\Omega_1)_{ii} = 1$ and $\Omega_{ij} = 0.3$ for $|i - j| = 1$. We set $\Omega_2 = \Omega_1 + \Omega$, where $\Omega$ is a symmetric and sparse matrix with $ \Omega_{10, ~10} = -0.3758, \Omega_{10, ~30} = 0.0616, \Omega_{10, ~50} = 0.2037, \Omega_{30, ~30} = -0.5482, \Omega_{30, ~50} = 0.0286$, and $\Omega_{50, ~50} = - 0.4614$. The other 3 nonzero entries in the lower triangle of $\Omega$ are determined by symmetry.
\item { Model 2:} We set $(\Omega_1)_{ij} = 0.5^{|i - j|}$ and let $\Omega_2 = \Omega_1 + \Omega$, where $\Omega = I_p$.
\item { Model 3:} $\Omega_1$ is the same as {Model 2} and $\Omega_2 = \Omega_1$.
\item { Model 4:} $\Omega_1$ is the same as {Model 2} and $\Omega$ is a band matrix defined as $(\Omega)_{ii} = 1$ and $(\Omega)_{ij} = 0.5$ for $|i - j| = 1$. Let $\Omega_2 = \Omega_1 + \Omega$.
\item {Model 5:} $\Omega_1 = I_p$ and $\Omega_2 = \Omega_1 + \Omega$ where $\Omega$ is a random sparse symmetric matrix with conditional number $10$ and non-zero density $n_1/p^2 \times 0.7$ (generated by {\em sprandsym} in Matlab).
\item {Models 6 and 7:} These are Cases 9 and 10 in \cite{Srivastava}. For Models 6 and 7, we generate the covariance by the following process: We first sample an $p\times p$ matrix $R_1$ where each entry is an i.i.d uniform random variable between $[0, 1]$. Then $\Omega_1$ is defined to be $\Omega_1 = (R_1^TR_1)^{-1}$ which is a dense matrix. Matrix $\Omega_2$ is generated similarly. For Model 6 the means are set as $u_1 = u_2 = 0$ while for Model 7, $u_1$ and $u_2$ are generated by random sampling over uniform distributions.
\item {Models 8 and 9:} These are Cases 10 and 11 in \cite{Srivastava} where we generate the same $p\times p$ matrix $R_1$ as in Model 6. Then $\Omega_1$ is defined to be $\Omega_1 = (R_1^TR_1R_1^TR_1)^{-1}$ representing the ellipsoidal covariances, which often have only one strong eigenvalue and many relatively smaller ones. Matrix $\Omega_2$ is similarly generated. The means are set as $u_1 = u_2 = 0$ for Model 8 and generated by random sampling for Model 9.
\end{itemize}
Model 1 is a model where $\Omega$ is a sparse two-block matrix after permutation. This is a model that favors {IIS-SQDA}. In Model 2, the difference between $\Omega_1$ and $\Omega_2$ is a diagonal matrix, and {IIS-SQDA} is expected to underperform as its screening step for identifying variables that are involved in interaction would retain all the variables. Model 3 is obviously a model that favors the linear discriminant analysis ({LDA}) as $\Omega=0$, and in particular favors the sparse LDA ({DSDA}). This model is simulated to test whether methods designed for sparse QDA work satisfactorily in situations where LDA works the best. In Model 4, the difference matrix $\Omega$ is a tridiagonal matrix where the screening step of {IIS-SQDA} is expected to underperform. Finally, in Model 5, $\Omega$ admits a random sparse structure having $0.7n_1=70$ nonzero entries regardless of the dimension $p$. In Model 6 - 9, the covariance matrices are dense and so are $\Omega$ and $\delta$. These cases are to test DA-QDA in scenarios where the sparse assumption fails to hold and the features are highly correlated. Model 6 and 8 are two difficult cases for linear methods as the means of the two classes are also the same. 
Our implementation of { IIS-SQDA} is applied only to { Models 2 - 9} while the results for {Model 1} are directly cited from \cite{Fan:2015aa}.
\par
 For models 1 - 5, we simulate 100 synthetic datasets for each model and record for each method under comparison: 1) The misclassification rate ({MR}), 2) The false positives for main effects and interactions ({FP.main} for $\delta$ and {FP.inter} for $\Omega$), 3) The false negatives for main effects and interactions ({FN.main} for $\delta$ and {FN.inter} for $\Omega$). The results are summarized in Tables \ref{table: select-M1}, \ref{table: select-M2}, \ref{table: select-M3}, \ref{table: select-M4} and \ref{table: select-M5} for the five models under consideration. For the dense models, we compare the misclassification rates over 100 replications, and the results are summarized in Table \ref{table: select-M6}.
 
 From these tables, we can observe the following phenomena.
 \begin{itemize}
 \item[1.] For Model 1 where the setup favors IIS-SQDA, IIS-SQDA performs the best in terms of variable selection. DA-QDA performs similarly in terms of the misclassification rate. These two methods outperform LDA and QDA (when $p=50$), and PLR, PLR2, DSDA by a large margin in classification error.
 \item[2.] For Model 2, as expected, IIS-SQDA is outperformed by DA-QDA by a large margin. Interestingly, PLR2 performs the second best. Linear classifiers including PLR and DSDA perform quite badly.
 \item[3.] For Model 3, DSDA, IIS-SQDA and DA-QDA perform best and similarly. It is interesting that DA-QDA performs on par with DSDA even when the model clearly favors sparse linear classifiers.
 \item[4.] For Model 4, DA-QDA outperforms all other methods again by a large margin. 
 \item[5.] For Model 5, DA-QDA performs the best and by a large margin when $p$ becomes large.
 \item[6.] For Models 6 - 9, the ordinary QDA performs the best for all low dimensional cases, which is as expected as the covariances of the the two classes are sufficiently different. When the dimension goes higher, DA-QDA achieves a high precision for some really difficult cases which is pretty surprising, considering the matrices are now all dense. The advantage mostly comes from that DA-QDA only imposes the sparse assumption on $\Omega$ and $\delta$ instead of the original precision matrices as for sQDA. In addition, DA-QDA performs better than IIS-SQDA.
 \end{itemize}
 To summarize, {DA-QDA} achieves the smallest misclassification rate in most examples and competitive performance in selecting main and interaction effects. IIS-SQDA is the preferred
approach if $\Omega$ is a two-block diagonal matrix after permutation as is the case for Model 1. PLR2 generally performs better than (sparse) linear classifiers when interactions exist.

\begin{table}[!htbp]
\scriptsize
\caption{The means and standard errors (in parentheses) of various
performance measures by different classification methods for model 1 based on 100 replications}
\label{table: select-M1}
 
\begin{tabular*}{\textwidth}{@{\extracolsep{\fill}}l|l|ccccc@{}}
\hline
${p}$ & {Method} & {MR (\%)} & {FP.main} & {FP.inter} & {FN.main} & \multicolumn{1}{c@{}}{{FN.inter}} \\
\hline
\hline
\phantom{0}50 & LDA & 39.43 (0.15) & -- & -- & -- & --\\
& QDA & 43.47 (0.10) & -- & -- & -- & -- \\
& PLR & 36.12 (0.26) & 5.95 (0.93) &-- &1.21 (0.04) &-- \\
& DSDA & 35.05 (0.22) & 8.81 (1.06) & -- & 0.07 (0.03) & --\\
& sQDA & 27.64 (0.22) & 11.17 (1.49) & -- & 0.33 (0.05) & -- \\
& PLR2 & 30.15 (0.44) & 0.51 (0.14) &11.26 (2.78) & 0.60 (0.05) & 2.62
(0.09)\\
& IIS-SQDA & 27.56 (0.27) & 5.60 (0.82) & 2.16 (0.32) & 0.19 (0.04) &
2.05 (0.09)\\
& DA-QDA & 26.50 (0.28) & 0.85 (0.18) & 35.26 (4.72) & 0.39 (0.07) & 3.74 (0.14)\\
& Oracle & 23.04 (0.09) & -- & -- & -- & --\\
\hline
200 
& PLR & 37.62 (0.34) & 7.82 (1.87) &-- &1.47 (0.05) &-- \\
& DSDA & 36.34 (0.30) & 15.06 (3.37) & -- & 0.36 (0.05) & --\\
& sQDA & 26.80 (0.21) & 12.75 (2.22) & -- & 0.47 (0.05) & -- \\
& PLR2 & 32.55 (0.53) & 0.25 (0.06) &17.44 (3.63) & 0.90 (0.05) & 2.72
(0.08)\\
& IIS-SQDA & 26.94 (0.31) & 6.43 (1.24) &0.78 (0.17) & 0.42 (0.05) &
2.22 (0.08)\\
& DA-QDA & 26.51 (0.20) & 0.29 (0.07) & 25.48 (2.75) & 0.82 (0.08) & 4.14 (0.12)\\
& Oracle & 21.93 (0.08) & -- & -- & -- & --\\
\hline
500 
& PLR & 38.82 (0.33) & 9.31 (1.99) &-- &1.58 (0.05) &-- \\
& DSDA & 37.10 (0.29) & 16.06 (3.02) & -- & 0.42 (0.05) & --\\
& sQDA & 28.22 (0.41) & 24.22 (5.04) & -- & 0.58 (0.05) & -- \\
& PLR2 & 35.45 (0.64) & 0.34 (0.09) &55.69 (12.67) & 0.99 (0.05) &
3.05 (0.10)\\
& IIS-SQDA & 26.78 (0.31) & 3.22 (1.09) &0.23 (0.05) & 0.98 (0.02) &
2.65 (0.09)\\
& DA-QDA & 26.68 (0.27) & 0.14 (0.06) & 10.96 (1.38) & 1.02 (0.08) &  4.36 (0.09)\\
& Oracle & 21.81 (0.09) & -- & -- & -- & --\\
\hline
\end{tabular*}
 
\end{table}

\begin{table}[!htbp]
\scriptsize
\caption{The means and standard errors (in parentheses) of various
performance measures by different classification methods for model 2
based on 100 replications}
\label{table: select-M2}
\begin{tabular*}{\textwidth}{@{\extracolsep{\fill}}l|l|ccccc@{}}
\hline
$p$ & {Method} & {MR (\%)} & {FP.main} & {FP.inter} & {FN.main} & \multicolumn{1}{c@{}}{{FN.inter}} \\
\hline
\hline
\phantom{0}50 & LDA & 34.53 (0.19) & -- & -- & -- & --\\
& QDA & 32.09 (0.25) & -- & -- & -- & -- \\
& PLR & 31.58 (0.20) & 7.51 (0.55) &-- &0.07 (0.03) &-- \\
& DSDA & 29.89 (0.16) & 8.52 (0.86) & -- & 0.16 (0.04) & --\\
& sQDA & 30.96 (0.90) & 27.33 (1.95) & -- & 0.24 (0.05) & --\\
& PLR2 & 5.85 (0.10) & 1.14 (0.11) &45.60 (1.08) & 0.14 (0.04) & 14.27 (0.33)\\
& IIS-SQDA & 5.85 (0.10) & 1.14 (0.11) &45.60 (1.08) &  0.14 (0.04) & 14.27 (0.32)\\
& DA-QDA & 1.84 (0.08) & 4.12 (0.49) & 110.10 (10.54) & 0.28 (0.05) & 1.28 (0.22)\\
& Oracle & 0.65 (0.02) & -- & -- & -- & -- \\
\hline
200 
& PLR & 33.34 (0.21) & 10.79 (0.70) &-- & 0.16 (0.04) &-- \\
& DSDA & 30.37 (0.23) & 11.91 (2.19) & -- & 0.29 (0.05) & --\\
& sQDA & 33.28 (0.58) & 101.75 (7.72) & -- & 0.27 (0.05) & --\\
& PLR2 & 1.73 (0.06) & 0.01 (0.01) &12.68 (0.56) & 1.08 (0.05) & 119.95 (0.52)\\
& IIS-SQDA & 3.98 (0.10) & 2.10 (0.15) &15.76 (0.60) & 0.11 (0.04) & 153.47 (0.31) \\
& DA-QDA & 0.39 (0.18) & 9.03 (2.12) & 724.35 (19.52) & 0.21 (0.04) & 6.05 (0.35)\\
& Oracle & 0 (0) & -- & -- & -- & --\\
\hline
500 
& PLR & 34.04 (0.24) & 11.17 (1.02) &-- &0.30 (0.05) &-- \\
& DSDA & 30.99 (0.22) & 14.61 (2.64) & -- & 0.44 (0.05) & --\\
& sQDA & 36.92 (0.64) & 243.9 (21.2) & -- & 0.35 (0.05) & --\\
& PLR2 & 1.68 (0.06) & 0 (0) & 5.52 (0.33) & 1.19 (0.05) & 401.47 (0.59)\\
& IIS-SQDA & 4.12 (0.09) & 2.74 (0.25) & 8.02 (0.43) & 0.12 (0.04) & 451.13 (0.29)\\
& DA-QDA & 0.16 (0.22) & 24.33 (2.18) & 4.81e3 (290.1) & 0.52 (0.05) & 58.09 (1.10)\\
& Oracle & 0 (0) & -- & -- & -- & --\\
\hline
\end{tabular*}
%
\end{table}

\begin{table}[!htbp]
\scriptsize
\caption{The means and standard errors (in parentheses) of various
performance measures by different classification methods for model 3
based on 100 replications}
\label{table: select-M3}
\begin{tabular*}{\textwidth}{@{\extracolsep{\fill}}l|l|ccccc@{}}
\hline
${p}$ & {Method} & {MR (\%)} & {FP.main} & {FP.inter} & {FN.main} & \multicolumn{1}{c@{}}{{FN.inter}} \\
\hline
\hline
\phantom{0}50 & LDA & 38.82 (0.19) & -- & -- & -- & --\\
& QDA & 47.57 (0.11) & -- & -- & -- & -- \\
& PLR & 36.06 (0.23) & 7.73 (0.58) &-- &0.14 (0.03) &-- \\
& DSDA & 34.82 (0.24) & 9.54 (1.09) & -- & 0.26 (0.04) & --\\
& sQDA & 41.52 (0.51) & 14.89 (1.69) & -- & 0.44 (0.05) & --\\
& PLR2 & 37.36 (0.34) & 0.60 (0.10) & 31.10 (3.21) & 0.39 (0.06) & 0 (0)\\
& IIS-SQDA & 35.10 (0.22) & 5.25 (0.46) &10.85 (0.96) &  0.06 (0.02) & 0 (0)\\
& DA-QDA & 34.99 (0.58) & 0.82 (0.20) & 23.84 (6.69) & 0.35 (0.07) & 0 (0)\\
& Oracle & 31.68 (0.10) & -- & -- & -- & --\\
\hline
200 
& PLR & 38.50 (0.31) & 12.90 (1.08) &-- & 0.23 (0.04) &-- \\
& DSDA & 36.27 (0.28) & 14.81 (2.26) & -- & 0.41 (0.05) & --\\
& sQDA & 43.82 (0.53) &53.18 (6.74) & -- & 0.52 (0.05) & --\\
& PLR2 & 40.31 (0.45) & 0.15 (0.05) & 40.38 (5.05) & 0.74 (0.06) & 0 (0)\\
& IIS-SQDA & 36.32 (0.25) & 25.39 (0.66) & 6.03 (0.50) & 0 (0) & 0 (0) \\
& DA-QDA & 36.55 (0.74) & 1.70 (1.38) & 37.15 (16.39) & 0.89 (0.09) & 0 (0)\\
& Oracle & 31.54 (0.10) & -- & -- & -- & --\\
\hline
500 
& PLR & 39.98 (0.32) & 14.79 (1.41) &-- &0.40 (0.05) &-- \\
& DSDA & 37.07 (0.29) & 19.49 (3.65) & -- & 0.59 (0.05) & --\\
& sQDA & 46.00 (0.48) & 130.91 (18.08) & -- & 0.57 (0.05) & --\\
& PLR2 & 42.23 (0.53) & 0.03 (0.02) & 36.6 (4.32) & 1.07 (0.06) & 0 (0)\\
& IIS-SQDA & 37.45 (0.26) & 14.53 (1.38) & 3.70 (0.32) & 0.07 (0.26) & 0 (0)\\
& DA-QDA & 37.95 (0.76) & 0.2 (0.06) & 57.49 (14.74) & 1.05 (0.09) & 0 (0)\\
& Oracle & 31.85 (0.12) & -- & -- & -- & --\\
\hline
\end{tabular*}
%
\end{table}

\begin{table}[!htbp]
\scriptsize
\caption{The means and standard errors (in parentheses) of various
performance measures by different classification methods for model 4
based on 100 replications}
\label{table: select-M4}
\begin{tabular*}{\textwidth}{@{\extracolsep{\fill}}l|l|ccccc@{}}
\hline
${p}$ & {Method} & {MR (\%)} & {FP.main} & {FP.inter} & {FN.main} & \multicolumn{1}{c@{}}{{FN.inter}} \\
\hline
\hline
\phantom{0}50
& LDA 		& 35.58 (0.20) & -- 			& -- 			& -- & --\\
& QDA 		& 35.40 (0.20) & -- 			& -- 			& -- & -- \\
& PLR 		& 32.42 (0.23) & 8.03 (0.57) 	&-- 			&0.03 (0.01) &-- \\
& DSDA 		& 31.39 (0.21) & 11.02 (1.13) 	& -- 			& 0.09 (0.03) & --\\
& sQDA      & 40.90 (0.46) & 18.36 (1.93)   &--             &0.45 (0.05)  &--\\
& PLR2 	  	& 22.42 (0.21) & 1.88 (0.16) 	& 81.56 (2.26) 	& 0.06 (0.03) & 123.72 (0.36)\\
& IIS-SQDA 	& 21.77 (0.20) & 3.42 (0.21) 	& 58.92 (1.86) 	&  0 (0)    & 125.73 (0.32)\\
& DA-QDA     	& 16.91 (0.27) & 0.55 (0.14) 	& 194.98 (11.31) 	& 0.61 (0.08)& 106.51 (0.83)\\
& Oracle    & 3.22 (0.04) & -- & -- & -- & --\\
\hline
200 
& PLR & 34.93 (0.28) & 12.71 (0.88) &-- & 0.10 (0.03) &-- \\
& DSDA & 32.64 (0.26) & 15.63 (2.14) & -- & 0.21 (0.04) & --\\
& sQDA & 41.68 (0.54) & 64.88 (7.33) & -- & 0.46 (0.05) & --\\
& PLR2 & 21.82 (0.20) & 0.30 (0.05) & 107.80 (2.32) & 0.40 (0.05) & 559.23 (0.63)\\
& IIS-SQDA & 20.15 (0.19) & 6.11 (0.31) & 70.76 (1.76) & 0 (0) & 563.33 (0.38) \\
& DA-QDA & 9.59 (0.19) & 0.31 (0.08) & 297.38 (25.33) & 0.82 (0.09) & 498.61 (1.49)\\
& Oracle & 0.28 (0.02) & -- & -- & -- & --\\
\hline
500 
& PLR & 37.19 (0.32) & 15.68 (1.27) &-- &0.32 (0.04) &-- \\
& DSDA & 33.83 (0.30) & 22.90 (3.54) & -- & 0.45 (0.05) & --\\
& sQDA & 43.39 (0.48) & 193.04 (20.32) & -- & 0.46 (0.05) &--\\
& PLR2 & 23.06 (0.23) & 0.05 (0.02) & 114.94 (2.34) & 0.79 (0.05) & 1455 (0.65)\\
& IIS-SQDA & 19.07 (0.17) & 12.86 (0.42) & 57.44 (1.41) & 0 (0) & 1459 (0.34)\\
& DA-QDA & 4.18 (0.13)  & 0.20 (0.04) & 298.24 (20.8) & 0.42 (0.07) & 1315 (2.41)\\
& Oracle & 0 (0) & -- & -- & -- & -- \\
\hline
\end{tabular*}
%
\end{table}

\begin{table}[!htbp]
\scriptsize
\caption{The means and standard errors (in parentheses) of various
performance measures by different classification methods for model 5
based on 100 replications}
\label{table: select-M5}
\begin{tabular*}{\textwidth}{@{\extracolsep{\fill}}l|l|ccccc@{}}
\hline
${ p}$ & {Method} & {MR (\%)} & {FP.main} & {FP.inter} & {FN.main} & \multicolumn{1}{c@{}}{{FN.inter}} \\
\hline
\hline
\phantom{0}50
& LDA 		& 39.21 (0.20) & -- 			& -- 			& -- & --\\
& QDA 		& 46.41 (0.17) & -- 			& -- 			& -- & -- \\
& PLR 		& 35.76 (0.26) & 6.08 (0.43) 	&-- 			& 0.01 (0.01) &-- \\
& DSDA 		& 33.73 (0.25) & 8.08 (0.99) 	& -- 			& 0.14 (0.04) & --\\
& sQDA      & 36.76 (0.27) & 9.37 (1.57)    & --            & 0.06 (0.02) & --\\
& PLR2 	  	& 36.62 (0.39) & 1.04 (0.13) 	& 45.83 (3.99) 	& 0.05 (0.02) & 63.69 (0.39)\\
& IIS-SQDA 	& 35.56 (0.29) & 8.77 (0.50) 	& 14.85 (0.83) 	&  0 (0)      & 61.18 (0.26)\\
& DA-QDA     	& 34.32 (0.53) & 0.52 (0.12)    & 39.76 (6.47)  & 0.58 (0.08) & 59.76 (0.64)\\
& Oracle    & 32.36 (0.25) & -- & -- & -- & --\\
\hline

200 
& PLR & 37.73 (0.34) & 9.68 (0.89) &-- & 0.40 (0.03) &-- \\
& DSDA & 34.58 (0.35) & 10.87 (2.44) & -- & 0.11 (0.03) & --\\
& sQDA & 26.11 (0.27) & 18.35 (4.79) & -- & 0.21 (0.04) & --\\
& PLR2 & 37.40 (0.44) & 0.32 (0.06) & 66.44 (5.47) & 0.31 (0.06) & 194.46 (0.35)\\
& IIS-SQDA & 33.22 (0.28) & 19.87 (0.93) & 6.16 (0.41) & 0 (0) & 191.37 (0.10) \\
& DA-QDA   & 29.35 (0.41) & 0.10 (0.05) & 164.24 (73.3) & 1.27 (0.07) & 175.8 (0.96)\\
& Oracle & 20.09 (0.27) & -- & -- & -- & -- \\
\hline
500 
& PLR & 39.13 (0.33) & 14.39 (1.29) &-- & 0.08 (0.03) &-- \\
& DSDA & 34.76 (0.25) & 9.44 (1.77) & -- & 0.16 (0.04) & --\\
& sQDA & 10.17 (0.16) & 22.32 (6.88) &-- & 0.24 (0.05) & --\\
& PLR2 & 37.44 (0.52) & 0.16 (0.05) & 90.78 (6.06) & 0.43 (0.06) & 493.48 (0.41)\\
& IIS-SQDA & 26.57 (0.23) & 19.14 (0.57) & 62.00 (1.56) & 0 (0) & 475.49 (0.20)\\
& DA-QDA &  23.75 (0.49) & 4.03 (2.91) & 507.92 (225.36) & 1.59 (0.06) & 459.96 (1.96)\\
& Oracle & 4.16 (0.08) & -- & -- & -- & --\\
\hline
\end{tabular*}
%
\end{table}

\begin{table}[!htbp]
\scriptsize
\caption{The means and standard errors (in parentheses) of mis-classification rate (MR \%) for models 6, 7, 8, 9
based on 100 replications}
\label{table: select-M6}
\begin{tabular*}{\textwidth}{@{\extracolsep{\fill}}l|l|cccc}
\hline
${ p}$ & {Method} & Model 6 & Model 7 & Model 8 & Model 9 \\
\hline
\hline
\phantom{0}50
& LDA     & 49.86 (0.10) & 21.91 (0.37) & 49.69 (0.10)     & 32.41 (0.42) \\
& QDA     & 0.00 (0.00) & 0.00 (0.00) & 0.00 (0.00)    & 0.00 (0.00) \\
& PLR     & 49.97 (0.07) & 24.40 (0.38)  & 49.77 (0.12)   & 40.24 (0.79)\\
& DSDA    & 49.89 (0.11) & 22.84 (0.39)  & 49.05 (0.23)    & 34.31 (0.50)\\
& sQDA      & 48.14 (0.37) & 32.80 (0.69) & 46.68 (0.35)  & 35.27 (0.68) \\
& PLR2      & 18.17 (0.24) & 16.85 (0.25)  & 7.41 (0.14)  & 7.28 (0.12) \\
& IIS-SQDA  & 19.90 (0.29) & 18.05 (0.25)  & 7.60 (0.13)  & 7.49 (0.12) \\
& DA-QDA    & 12.58 (0.23) & 11.70 (0.28)  & 4.26 (0.14)  & 4.22 (0.14) \\
& Oracle    &  0.00 (0.00)& 0.00 (0.00) & 0.00 (0.00) & 0.00 (0.00) \\
\hline

200 
& PLR & 50.08 (0.09) & 39.86 (0.39) & 49.66 (0.16) & 49.83 (0.08)\\
& DSDA & 49.96 (0.11) & 37.00 (0.28) & 49.33 (0.18) & 48.33 (0.34)\\
& sQDA & 50.48 (0.15) & 45.33 (0.45) & 48.98 (0.26) & 48.49 (0.20)\\
& PLR2 & 38.76 (0.43) & 36.90 (0.48) & 6.80 (0.13) & 7.25 (0.10)\\
& IIS-SQDA & 45.59 (0.37) & 40.93 (0.44) & 8.03 (0.11) & 8.50 (0.14)\\
& DA-QDA   & 31.45 (0.42) & 33.01 (0.50) & 4.65 (0.88) & 4.07 (0.67)\\
& Oracle &  0.00 (0.00) & 0.00 (0.00) & 0.00 (0.00) & 0.00 (0.00) \\
\hline
500 
& PLR & 49.90 (0.06) & 47.67 (0.26) & 50.15 (0.09) &  50.13 (0.10) \\
& DSDA & 50.22 (0.21) & 42.71 (0.14) & 49.38 (0.09) &  46.70 (0.16\\
& sQDA & 50.17 (0.08) & 50.33 (0.11) & 49.85 (0.12) & 49.83 (0.12)\\
& PLR2 & 46.32 (0.49) & 45.58 (0.46) & 5.66 (0.12) & 7.47 (0.06)\\
& IIS-SQDA & 49.33 (0.44) & 46.61 (0.39) & 7.37 (0.15) & 9.43 (0.13) \\
& DA-QDA &  39.67 (0.16) & 41.88 (0.48) & 2.43 (0.12) & 2.27 (0.07) \\
& Oracle &  0.00 (0.00) & 0.00 (0.00) & 0.00 (0.00) & 0.00 (0.00) \\
\hline
\end{tabular*}
%
\end{table}

\subsection{Real data}

In this section, we investigate the performance of {DA-QDA} by analyzing four real data sets and compare it to the other classifiers discussed in the simulation study and  we also include L1-regularized SVM (L-SVM) and kernel SVM (K-SVM, with Gaussian kernel) for more comprehensive comparison.

\paragraph{Quora answer classifier.} This is a data challenge available at \url{http://www.quora.com/challenges#answer_classifier}. The training data set contains 4,500 answers from QUORA which have been annotated with either "good" or "bad". For each answer, 21 features (20 of which are effective) were extracted from the original sentences. The goal of this challenge is to automatically classify a new answer based on the 20 features. Since the dimension $p=20$ is relatively small, we can compare DA-QDA to all the methods discussed in the simulation via 10-fold cross-validation. In particular, we randomly split the data into ten parts, fit a model to the nine parts of the data, and report the misclassification error on the part that is left out. The average misclassification errors and the standard errors for various methods are in Table \ref{tab:quora}. Interestingly, LDA performs much better than QDA, suggesting that if we stop the analysis here, we might simply have preferred to use the linear classifier LDA. However, the story becomes different if sparse models are considered. In particular, PLR, PLR2, IIS-SQDA and DA-QDA all outperform the non-sparse models significantly with DA-QDA performing the best.

\begin{table}[!htpb]
\caption{Misclassification rate (\%) for the Quora answer data under 10-fold cross-validation}
\label{tab:quora}
\centering
\begin{tabular}{lcc}
\hline
{Method}  & {mean} & {standard error} \\
\hline
 DA-QDA       & 16.44 & 0.45  \\
 LDA 		& 18.84           & 0.50 		\\
 QDA 		& 30.33           & 0.72	  \\
 PLR 		& 17.89 		  & 0.60  \\
 DSDA 		& 19.11 		  & 0.56 \\
 sQDA       & 29.59           & 1.57 \\
 PLR2 	  	& 17.56 		  & 0.71  \\
 IIS-SQDA 	& 17.33			  & 0.48  \\
 L-SVM  & 18.13          &0.83  \\
 K-SVM & 25.53        & 0.27  \\
\hline
\end{tabular}
\end{table}

\paragraph{Gastrointestinal Lesions} 
This dataset \citep{PMesejo:etal:2016} contains the features extracted from a database of colonoscopic videos showing three types of gastrointestinal lesions, hyperplasic, adenoma and serrated adenoma. The original task is a multi-class classification problem, which is simplified to a binary classification task aiming at identifying adenoma. The data set contains 152 samples (76 original samples each with two different light conditions) and 768 features. We select the top 200 features with the largest absolute values of the two sample $t$ statistics and perform a 10-fold corss-validation. The average misclassification errors and the standard errors are reported in Table \ref{tab:lensions}. The data is predominated by the main effects as the logistic regression achieves the best with DA-QDA as the runner-up.

\begin{table}[!htpb]
\caption{Misclassification rate (\%) for gastrointestinal lesions under 10-fold cross-validation}
\label{tab:lensions}
\centering
\begin{tabular}{lcc}
\hline
{Method}  & {mean} & {standard error} \\
\hline
 DA-QDA       & 33.33  & 0.00  \\
 LDA    & ---           & ---    \\
 QDA    & ---           & ---    \\
 PLR    & 30.67       & 2.04  \\
 DSDA     & 53.33       & 0.00 \\
 sQDA       & 40.00           & 0.00 \\
 PLR2       & 39.33       & 0.67  \\
 IIS-SQDA   & 40.00       & 1.02  \\
 L-SVM        & 44.00       & 1.85  \\
 k-SVM      & 46.67       & 0.00  \\
\hline
\end{tabular}
\end{table}

\paragraph{Pancreatic cancer RNA-seq data} 
The dataset \citep{Weinstein:etal:2013} is part of the RNA-Seq (HiSeq) PANCAN data set and is a random extraction of gene expressions of patients having different types of tumor: BRCA, KIRC, COAD, LUAD and PRAD. The dataset contains 801 patients and 20531 genes. In this task, we aim to distinguish BRCA against the other cancers. Similar to the previous study, we select 500 genes with the largest absolute values of the two sample $t$ statistics for further analysis. Since most methods achieve 0 misclassification error in 10-fold cross-validation test, to increase the difficulty, we randomly split the dataset in two equal subsets, train on one subset and test on the other. We repeat this procedure 50 times to obtain the following Table \ref{tab:pancan}. The L1 regularized SVM achieves the smallest misclassification error among all the methods. Similar as the previous dataset, the difference between cancers is dominated by main factors as the LDA has already achieved a surprisingly low misclassification rate (0.22\%). The dataset also shows that DA-QDA can also perform well when the difference between the covariance matrices is small.

\begin{table}[!htpb]
\caption{Misclassification rate (\%) for Pancreatic cancer RNA-seq data}
\label{tab:pancan}
\centering
\begin{tabular}{lcc}
\hline
{Method}  & {mean} & {standard error} \\
\hline
 DA-QDA       & 0.09  & 0.13  \\
 LDA    & 0.22           & 0.22    \\
 QDA    & ---           & ---    \\
 PLR    & 0.12       & 0.01  \\
 DSDA     & 0.62     & 0.22 \\
 sQDA       & 0.29           & 0.19 \\
 PLR2       & 0.09       & 0.01  \\
 IIS-SQDA   & 0.15       & 0.01  \\
 L-SVM        & 0.02       & 0.01  \\
 k-SVM      & 37.26      & 0.48  \\
\hline
\end{tabular}
\end{table}

\paragraph{Prostate cancer}
Taken from \url{ftp://stat.ethz.ch/Manuscripts/dettling/prostate.rda}, 
this data contains genetic expression levels for $N = 6033$ genes of $102$ individuals. The first $50$ are normal control subjects while the rest are prostate cancer patients. More details of the data can be found in \cite{Singh}, \cite{Dettling} and \cite{Efron:2010}. The goal is to identify genes that are linked with prostate cancer and predict potential patients and the difficulty of this task lies in the interactions among genes. The existence of interactions can often complicate the analysis and produce unreliable inference if they are ignored. For example, Figure \ref{fig:2} displays the pair of $118^{th}$ and $182^{th}$ gene. We can see the marginal distributions of each gene does not differ too much between the patients and the normal subjects (the middle and the right panels), suggesting that their main effects may not be important for distinguishing the two classes. In the left panel of Figure \ref{fig:2}, however, we can identify some joint pattern that distinguishes the two groups. It can be seen that most patients are allocated in the red triangle while most normal subjects are within the blue triangle, indicating the existence of some interaction effect that might be useful for classification.
\begin{figure}[!htbp]
\centering
 \includegraphics[height = 5cm]{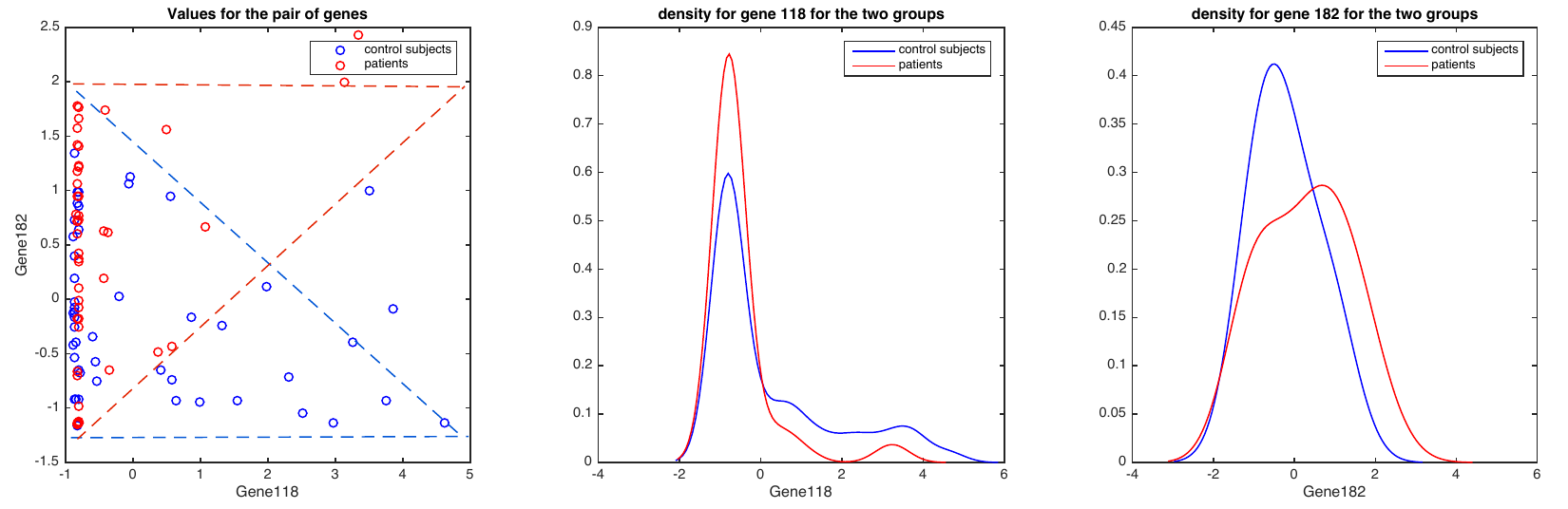}
\caption{The plot for the gene 118 and gene 182. Left: joint scatter plot; Middle: marginal density of gene 118; Right: marginal density of gene 182.}
\label{fig:2}
\end{figure}

For this data, we follow the same method in \cite{Cai}, retaining only the top 200 or 500 genes with the largest absolute values of the two sample $t$ statistics.  The average misclassification errors and the standard errors using 10-fold cross-validation for various methods are reported in Table \ref{tab:prostate}. Note that since $p\gg n$,  LDA and QDA were excluded. We can see again that DA-QDA is on par with L1 regularized SVM and outperforms all the other methods by a large margin, regardless of the number of the genes that were used for analysis.

\begin{table}[!htpb]
\caption{Misclassification rate (\%) for the prostate cancer data under 10-fold cross-validation}
\label{tab:prostate}
\centering
\begin{tabular}{lccccc}
\hline
& \multicolumn{2}{c}{$p = 200$} && \multicolumn{2}{c}{$p = 500$}\\
\cline{2 - 3}\cline{5 - 6}
{Method}  & {mean} & {std error} & & {mean} & {std error} \\
\hline
 DA-QDA       & 0.00   & 0.00      &&  1.00   & 1.00\\
 LDA    & ---           & ---      &&  ---    & --- \\
 QDA    & ---           & ---      &&  ---    & --- \\
 PLR 		& 11.00		      & 2.45      &&  16.00  & 2.92\\
 DSDA 		& 5.00 		      & 3.32      &&  11.00  & 2.92\\
 sQDA       & 0.00            & 0.00      &&  2.00   & 2.00\\
 PLR2 	  	& 26.00 	      & 4.47      &&  43.00  & 3.39\\
 IIS-SQDA 	& 11.00			  & 2.92      &&  18.00  & 2.74\\
 L-SVM      & 0.00        & 0.00    &&   1.00    & 1.00\\
 k-SVM      & 48.00       & 9.82    &&   63.00   & 2.55\\
\hline
\end{tabular}
\end{table}

\section{Conclusion}
We have proposed a novel method named DA-QDA for high-dimensional quadratic discriminant analysis. This is the first method aiming at directly estimating the quantities in the QDA discriminant function. The proposed framework is simple, fast to implement and enjoys excellent theoretical properties. We have demonstrated via extensive simulation and four data analyses that DA-QDA performs competitively under various circumstances.

We conclude by identifying three directions for future research. First, though the discussion of the paper is focused on binary problems, we can extend DS-QDA to handle multi-class problems as follows. When there are $k\geq 2$ classes, we can apply the DA-QDA approach to classes 1 and $j$, where $j=2,\ldots, k$, in a pairwise manner. For a new sample $z$, denote the  DA-QDA classifier between class  1 and class $j$ as $\hat{D}_j(z)$ and suppose $D_i(z)$ is the smallest among $\{D_{j}(z), j=2,\ldots, k\}$.
By Bayes' rule, we can then classifier $z$ into class $i$ if $D_i(z)>0$ and class $1$ otherwise.  
 Second, it is also interesting to see whether our theoretical results are optimal and in what sense. Finally, the proposed framework is extremely flexible. As a concrete example, if $\Omega$ is a two block sparse matrix  after permutation as in \cite{Fan:2015aa}, we can change the penalty $\|\Omega\|_1$ in \eqref{eq:DA-QDA} to one that encourages row sparsity, for example to $\|\Omega\|_{1,2}=\sum_{j=1}^p \|\Omega_{i,:}\|_2$ which is the sum of the $\ell_2$ norms of the rows. It will be interesting to see how well this procedure compares with IIS-SQDA in \cite{Fan:2015aa}. This topic is beyond the scope of the current paper and will be pursued elsewhere.





\section*{Acknowledgements}
We would like to thank the two reviewers and the Action Editor Prof. Saharon Rosset for their valuable comments that have greatly improved the manuscript. Jiang's research is supported by the Hong Kong RGC grant (PolyU 253023/16P). Leng's research is partially supported by the Alan Turing Institute under the EPSRC grant EP/N510129/1. 
 
\appendix
\section{Appendix A. Technical Lemmas and Proofs}
\label{app:theorem}



\subsection{Linear convergence of the ADMM algorithm}
The following lemma establishes the linear convergence of our proposed ADMM algorithm in solving \eqref{eq:DA-QDA1}.  

\begin{Lem}\label{linearADMM}
Given $\hat{\Sigma}_1,\hat{\Sigma}_2$ and $\lambda$, suppose that the ADMM scheme in \eqref{eq:Lambda}-\eqref{eq:Gamma} generates a solution sequence $\{\Omega^r,\Psi^r,\Lambda^r\}$. We have that $\{\Omega^r,\Psi^r\}$ converges linearly to an optimal solution of \eqref{eq:DA-QDA1}, and that $\|\Omega^r-\Psi^r\|_F$ converges linearly to zero. 

\end{Lem}
\begin{proof}
 Note that equation \eqref{eq:DA-QDA1} is a special case of (1.1) of \cite{Hong2017} with two blocks: $f(\Omega,\Psi)=f_1(\Omega)+f_2(\Psi)$ where, $f_1(\Omega)=\frac{1}{2}Tr\left(\Omega^T\hat\Sigma_1\Omega \hat\Sigma_2\right) -Tr\left( \Omega(\hat\Sigma_1-\hat\Sigma_2) \right)$ and 
 $f_2(\Psi)=\lambda\| \Psi\|_1$. Note that $Tr\left(\Omega^T\hat\Sigma_1\Omega \hat\Sigma_2\right)=vec(\Omega)^T(\hat{\Sigma}_2\otimes \hat{\Sigma}_1)vec(\Omega)$. Let $\hat{\Sigma}_2\otimes \hat{\Sigma}_1=U^T\Lambda U$ be the eigenvalue decomposition of the symmetric matrix $\hat{\Sigma}_2\otimes \hat{\Sigma}_1$, and denote $A_1:=U^T\Lambda^{1/2} U$. Let $g_1(x)=x^Tx$ be a function  defined on $R^{p^2}\mapsto R$, and $h_1(x)=tr(x(\hat\Sigma_1-\hat\Sigma_2))$, $h_2=\lambda|x|_1$ be functions defined on $R^{p^2}\mapsto R$. We then have $f_1(\Omega)=g_1(A_1vec(\Omega))+h_1(vec(\Omega))$ and $f_2(\Psi)=h_2(vec(\Psi))$. Clearly given $\hat{\Sigma}_1,\hat{\Sigma}_2$ and $\lambda$, the gradient of $g_1$ is  uniformly Lipschitz continuous and $h_1$, $h_2$ are both polyhedral. The lemma follows immediately from Theorem 3.1 of \cite{Hong2017}.

\end{proof}

\subsection{Proofs of Theorem \ref{thmomega}}

We first introduce some technical lemmas and the proof of Theorem \ref{thmomega} will be given after these lemmas.
\begin{Lem}\label{concentration}
Suppose $\lambda_{\rm max}(\Sigma_k)<\epsilon_0<\infty$ for $k=1,2$. There exist constants $C_0, C_1, C_2>0$ depending on $\epsilon_0$ only, such that for any $|v|\leq C_0$,
\begin{eqnarray*}
P(|\hat{\sigma}_{kij}-\sigma_{kij}|>v)\leq C_1\exp(-C_2(n_k-1)v^2)\leq  C_1\exp(-C_2nv^2).
\end{eqnarray*}
\end{Lem}
\begin{proof}
Denote $X=(X_1,\ldots,X_{n_1})^T$. Let $\Lambda$ be an orthogonal matrix with the last row being $(n^{-1/2},\ldots,n^{-1/2})$ and define $Z=(z_1,\ldots,z_n)=\Lambda X$. We then have $z_1,\ldots,z_{n-1}\sim N({\bf 0},\Sigma)$ and they are independent to each other. Note that $n_1\hat{\Sigma}_1=X^T(I_p-{\bf 1}{\bf 1}^T)X=Z^T\Lambda(I_p-{\bf 1}{\bf 1}^T)\Lambda^TZ=\sum_{i=1}^{n-1}z_i z_i^T$. This together with Lemma A.3 of \cite{Bickel:2008aa} prove Lemma \ref{concentration}.
\end{proof}
{\bf Remark.} Denote $\sigma^2=\max\{\sigma_{1ii},\sigma_{2ii},i=1,\ldots,p\}$. From Lemma 1 of \cite{Ravikumar:2011aa} we can see that Lemma \ref{concentration} is true for $C_1=4$, $C_2=[128(1+4\sigma^2)^2\sigma^4]^{-1}$ and $v=8(1+4\sigma^2)\sigma^2$.

\begin{Lem}\label{gammabound}
Assume that,
\begin{eqnarray}\label{assumpBGamma}
 B_{\Gamma,\Gamma^T}<\frac{1}{3(d^2\epsilon_1\epsilon_2+B d^2(\epsilon_1+\epsilon_2))}.
\end{eqnarray}\label{B0}
Let $R(\Delta_{\Gamma})=(\Gamma_{S,S}+\Delta_\Gamma)^{-1}-\Gamma_{S,S}^{-1}+\Gamma_{S,S}^{-1}(\Delta_\Gamma)_{S,S}\Gamma_{S,S}^{-1}$. We then have
  \begin{eqnarray}\label{B1}
||R(\Delta_{\Gamma})||_{\infty}\leq 3||(\Delta_{\Gamma})_{S,S}||_\infty||(\Delta_{\Gamma^T})_{S,S}||_{1,\infty}B_\Gamma B_{\Gamma^T}^2,
\end{eqnarray}
and
\begin{eqnarray}\label{B2}
||R(\Delta_{\Gamma})||_{1,\infty}\leq 3||(\Delta_{\Gamma})_{S,S}||_{1,\infty}||(\Delta_{\Gamma^T})_{S,S}||_{1,\infty}B_\Gamma B_{\Gamma^T}^2.
\end{eqnarray}
Moreover, we also have
\begin{eqnarray}\label{B3}
||\hat{\Gamma}_{S,S}^{-1}-\Gamma_{S,S}^{-1}||_{\infty}\leq 3d^2(\epsilon^2+2B\epsilon)^2B_\Gamma B^2_{\Gamma^T}+(\epsilon^2+2B\epsilon)B_\Gamma^2,
\end{eqnarray}
\begin{eqnarray}\label{B4}
||\hat{\Gamma}_{S,S}^{-1}-\Gamma_{S,S}^{-1}||_{1,\infty}\leq 3d^4(\epsilon^2+2B\epsilon)^2B_\Gamma B^2_{\Gamma^T}+d^2(\epsilon^2+2B\epsilon)B_\Gamma^2.
\end{eqnarray}
\end{Lem}
\begin{proof}
Note that
\begin{eqnarray*}
\Delta_\Gamma=\Delta_2\otimes \Delta_1 +\Delta_2\otimes \Sigma_1+\Sigma_2\otimes \Delta_1.
\end{eqnarray*}
Consequently by (\ref{assumpBGamma}) we have $||\Gamma_{S,S}^{-1}||_{1,\infty}||(\Delta_{\Gamma})_{S,S}||_{1,\infty}\leq 1/3$. (\ref{B1}) and (\ref{B2}) can then be proved using the same arguments as in Appendix B of \cite{Ravikumar:2011aa}. Note that 
$$||(\Delta_\Gamma)_{S,S}||_\infty\leq \epsilon^2+2B\epsilon,$$ and $$\max\{||(\Delta_{\Gamma})_{S,S}||_{1,\infty},||(\Delta_{\Gamma^T})_{S,S}||_{1,\infty}\}\leq d^2(\epsilon^2+2B\epsilon),$$ 
we have
\begin{eqnarray*}
||\hat{\Gamma}_{S,S}^{-1}-\Gamma_{S,S}^{-1}||_{\infty}&\leq& ||R(\Delta_{\Gamma})||_{\infty}+||\Gamma_{S,S}^{-1}(\Delta_\Gamma)_{S,S}\Gamma_{S,S}^{-1}||_\infty \\
&\leq& 3d^2(\epsilon^2+2B\epsilon)^2B_\Gamma B^2_{\Gamma^T}+ ||(\Delta_\Gamma)_{S,S}||_\infty ||\Gamma_{S,S}^{-1}||_{1,\infty}^2\\
&\leq&3d^2(\epsilon^2+2B\epsilon)^2B_\Gamma B^2_{\Gamma^T}+(\epsilon^2+2B\epsilon)B_\Gamma^2.
\end{eqnarray*}
This proves (\ref{B3}).  (\ref{B4}) can be proved similarly.
\end{proof}

\begin{Lem}\label{Omegalem}
Assume that (\ref{assumpBGamma}) and the following assumptions hold:  $\alpha>0$, $\epsilon< \min\Big\{B,\frac{\alpha\lambda}{2(2-\alpha)}\Big\}$ and
\begin{eqnarray}\label{assumpLemma3}
&& 3d^2\epsilon BB_{\Gamma,\Gamma^T} [1+(B_\Sigma^2+3d^2\epsilon B B_{\Gamma,\Gamma^T})(9d^2\epsilon B  B_{\Gamma,\Gamma^T}+1)B_{\Gamma,\Gamma^T}]\nonumber \\&&\leq C_{\alpha}\alpha\min\{\lambda,1\}  
\end{eqnarray}
where $C_\alpha=\frac{\alpha\lambda+2\epsilon\alpha-4\epsilon }{ 2B \alpha\lambda+\alpha\lambda+2\epsilon\alpha}$.
We have:

(i) $vec(\hat{\Omega})_S=0$.

(ii) $||\hat{\Omega}-\Omega||_\infty< 2\lambda B_{\Gamma,\Gamma^T}+9d^2\epsilon BB_{\Gamma,\Gamma^T}^2(3d^2\epsilon BB_{\Gamma,
\Gamma^T}+1) (2B+2\lambda).$
\end{Lem}
\begin{proof}
(i) Suppose $\tilde{\Omega}$ is the solution of:
\begin{equation}\label{*}
\tilde{\Omega}=\min_{\Omega\in R^{p\times p},\Omega_{S^c}=0} ~\frac{1}{2}Tr\left(\Omega^T\hat\Sigma_1\Omega \hat\Sigma_2\right) -Tr\left( \Omega(\hat\Sigma_1-\hat\Sigma_2) \right)+\lambda\| \Omega \|_1.
\end{equation}
We prove Lemma \ref{Omegalem} (i) by showing that $\hat{\Omega}=\tilde{\Omega}$. Due to the convexity of (3) in the main paper, we only need to show that the derivative of (3) is zero at $\tilde{\Omega}$. Equivalently, we need to show that for any $1\leq i,j\leq p$ we have,
\begin{eqnarray}\label{**}
|\hat{\Sigma}_1\tilde{\Omega}\hat{\Sigma}_2-(\hat{\Sigma}_1-\hat{\Sigma}_2)|_{i,j}\leq \lambda.
\end{eqnarray}
By taking the first derivative of (\ref{*}) we obtain,
\begin{eqnarray}\label{***}
\{\hat{\Sigma}_1\tilde{\Omega}\hat{\Sigma}_2-(\hat{\Sigma}_1-\hat{\Sigma}_2)+\lambda Z\}_S={\bf 0},
\end{eqnarray}
where $Z=(Z_{ij})_{1\leq i,j\leq p}$ with $Z_{ij}=0$ for $(i,j)\in S^c$,  $Z_{ij}={\rm sign} (\tilde{\Omega}_{ij})$ for $(i,j)\in S$ and $\tilde{\Omega}_{ij}\neq 0$, $Z_{ij}\in [-1,1]$ for $(i,j)\in S$ and $\tilde{\Omega}_{ij}=0$. Therefore (\ref{**}) is true for any $(i,j)\in S$. Using the vector operator, (\ref{***}) becomes
\begin{eqnarray*}
\{(\hat\Sigma_2 \otimes \hat\Sigma_1) vec(\tilde{\Omega})-vec(\hat\Sigma_1-\hat\Sigma_2)+\lambda vec(Z)\}_S={\bf 0}.
\end{eqnarray*}
Equivalently we have
\begin{eqnarray}\label{OmegaS}
vec(\tilde{\Omega})_S=\hat{\Gamma}_{S,S}^{-1}[vec(\hat\Sigma_1-\hat\Sigma_2)_S-\lambda vec(Z)_S ].
\end{eqnarray}
Note that the left hand side of (\ref{**}) equals $|\hat{\Sigma}_1\tilde{\Omega}\hat{\Sigma}_2-\Sigma_1+\Sigma_2-(\Delta_1-\Delta_2)|_{i,j}$. Using the vector operator and the fact that $\tilde{\Omega}_{S^c}={\bf 0}$, to show that (\ref{**}) is true for any $e\in S^c$, we only need to show that
\begin{eqnarray}\label{Sc1}
|\hat{\Gamma}_{e,S}vec(\tilde{\Omega})_S-\Gamma_{e,S}vec(\Omega)_S-vec(\Delta_1-\Delta_2)_e|\leq \lambda.
\end{eqnarray}
Here we have use the fact that $\Gamma_{e,S}vec(\Omega)_S=vec(\Sigma_1-\Sigma_2)_e$.
By (\ref{OmegaS}) and the fact that $vec(\Omega)_S=\Gamma_{S,S}^{-1}vec(\Sigma_1-\Sigma_2)_S$ we have,
\begin{eqnarray*}
&&|\hat{\Gamma}_{e,S}vec(\tilde{\Omega})_S-\Gamma_{e,S}vec(\Omega)_S-vec(\Delta_1-\Delta_2)_e| \\
&=&|\hat{\Gamma}_{e,S}\hat{\Gamma}_{S,S}^{-1}[vec(\hat\Sigma_1-\hat\Sigma_2)_S-\lambda vec(Z)_S ]\\
&&-\Gamma_{e,S}\Gamma_{S,S}^{-1}vec(\Sigma_1-\Sigma_2)_S-vec(\Delta_1-\Delta_2)_e| \\
&\leq& |[\hat{\Gamma}_{e,S}\hat{\Gamma}_{S,S}^{-1}-\Gamma_{e,S}\Gamma_{S,S}^{-1}]vec(\Sigma_1-\Sigma_2)_S|
\\
&&+ |\hat{\Gamma}_{e,S}\hat{\Gamma}_{S,S}^{-1}vec(\hat{\Sigma}_1-\hat{\Sigma}_2-\Sigma_1+\Sigma_2)_S|  +||\Delta_1-\Delta_2||_\infty+\lambda |\hat{\Gamma}_{e,S}\hat{\Gamma}_{S,S}^{-1}|_1\\
&\leq&2B |\hat{\Gamma}_{e,S}\hat{\Gamma}_{S,S}^{-1}-\Gamma_{e,S}\Gamma_{S,S}^{-1}|_1
+(\epsilon_1+\epsilon_2+\lambda) |\hat{\Gamma}_{e,S}\hat{\Gamma}_{S,S}^{-1}|_1+
\epsilon_1+\epsilon_2.
\end{eqnarray*}
Consequently, (\ref{Sc1}) is true if

\begin{equation}\label{Sc2}
\begin{array}{r@{}l}
    &\max_{e\in S^c}2B |\hat{\Gamma}_{e,S}\hat{\Gamma}_{S,S}^{-1}-\Gamma_{e,S}\Gamma_{S,S}^{-1}|_1+2\epsilon (1+|\hat{\Gamma}_{e,S}\hat{\Gamma}_{S,S}^{-1}|_1)\leq (1-C_\alpha)\alpha\lambda,\\
 & \max_{e\in S^c}|\hat{\Gamma}_{e,S}\hat{\Gamma}_{S,S}^{-1}|_1\leq 1-(1-C_\alpha)\alpha.
\end{array}
\end{equation}

Next we finish this proof by showing that (\ref{Sc2}) is true under the assumptions of this lemma.

By Lemma \ref{gammabound} we have for any $e\in S^c$,
\begin{eqnarray*}
&&|\hat{\Gamma}_{e,S}\hat{\Gamma}_{S,S}^{-1}-\Gamma_{e,S}\Gamma_{S,S}^{-1}|_1 \\
&\leq& |(\hat{\Gamma}_{e,S}-\Gamma_{e,S})\Gamma_{S,S}^{-1}|_1
+|\Gamma_{e,S}(\hat{\Gamma}_{S,S}^{-1}-\Gamma_{S,S}^{-1})|_1
\\&&+|(\hat{\Gamma}_{e,S}-\Gamma_{e,S})(\hat{\Gamma}_{S,S}^{-1}-\Gamma_{S,S}^{-1})|_1 \\
&\leq& |\hat{\Gamma}_{e,S}-\Gamma_{e,S}|_1 ||\Gamma_{S,S}^{-1}||_{1,\infty}+
||\hat{\Gamma}_{S,S}^{-1}-\Gamma_{S,S}^{-1}||_{1,\infty} B_\Sigma^2
\\
&&+|\hat{\Gamma}_{e,S}-\Gamma_{e,S}|_1 ||\hat{\Gamma}_{S,S}^{-1}-\Gamma_{S,S}^{-1}||_{1,\infty} \\
&\leq&d^2 (\epsilon^2+2B\epsilon)B_{\Gamma,\Gamma^T}+d^2[B_\Sigma^2+d^2(\epsilon^2+2B\epsilon)]\\
&&\times[3d^2(\epsilon^2+2B\epsilon)B_{\Gamma,\Gamma^T}+1](\epsilon^2+2B\epsilon)B_{\Gamma,\Gamma^T}^2\\
&\leq&3d^2\epsilon BB_{\Gamma,\Gamma^T} [1+(B_\Sigma^2+3d^2\epsilon B B_{\Gamma,\Gamma^T})(9d^2\epsilon B  B_{\Gamma,\Gamma^T}+1)B_{\Gamma,\Gamma^T}]\\
&\leq& C_\alpha \alpha\min\{\lambda,1\}.
\end{eqnarray*}
Consequently, $\max_{e\in S^c}|\hat{\Gamma}_{e,S}\hat{\Gamma}_{S,S}^{-1}|_1\leq C_\alpha\alpha\min\{\lambda,1\}+(1-\alpha) \leq 1-(1-C_\alpha)\alpha$, and
\begin{eqnarray*}
&&\max_{e\in S^c}2B  |\hat{\Gamma}_{e,S}\hat{\Gamma}_{S,S}^{-1}-\Gamma_{e,S}\Gamma_{S,S}^{-1}|_1+2\epsilon (1+|\hat{\Gamma}_{e,S}\hat{\Gamma}_{S,S}^{-1}|_1)\\
&\leq& 2B  C_\alpha \alpha\lambda+2\epsilon[2-(1-C_\alpha)\alpha]\\
&=& (1-C_\alpha)\alpha\lambda .
\end{eqnarray*}

(ii) From (i) we have $\hat{\Omega}=\tilde{\Omega}$. By (\ref{OmegaS}) and the fact that $vec(\Omega)_S=\Gamma_{S,S}^{-1}vec(\Sigma_1-\Sigma_2)_S$, we have
\begin{eqnarray*}
||\hat{\Omega}-\Omega||_\infty&=&|vec(\tilde{\Omega})-vec(\Omega)|_{\infty} \\
&=&|\hat{\Gamma}_{S,S}^{-1}[vec(\hat\Sigma_1-\hat\Sigma_2)_S-\lambda vec(Z)_S ]-\Gamma_{S,S}^{-1}vec(\Sigma_1-\Sigma_2)_S|_\infty\\
&\leq& \lambda||\hat{\Gamma}_{S,S}^{-1}||_{1,\infty}+|\hat{\Gamma}_{S,S}^{-1}vec(\Delta_1-\Delta_2)_S|_\infty
\\
&&+|(\hat{\Gamma}_{S,S}^{-1}-\Gamma_{S,S}^{-1})vec(\Sigma_1-\Sigma_2)|_\infty \\
&\leq&(\lambda+2\epsilon)||\hat{\Gamma}_{S,S}^{-1}||_{1,\infty}
+2B||\hat{\Gamma}_{S,S}^{-1}-\Gamma_{S,S}^{-1}||_{1,\infty}\\
&\leq&(\lambda+2\epsilon)||\Gamma_{S,S}^{-1}||_{1,\infty}+(2B+\lambda+2\epsilon)||\hat{\Gamma}_{S,S}^{-1}-\Gamma_{S,S}^{-1}||_{1,\infty}.
\end{eqnarray*}
By (\ref{B4}) and the assumption that $2\epsilon<\alpha\lambda/(2-\alpha)<\alpha\lambda<\lambda$ we immediately have
\begin{eqnarray*}
&&||\hat{\Omega}-\Omega||_\infty \\&\leq& (\lambda+2\epsilon)B_{\Gamma} + 3[d^4(\epsilon^2+2B\epsilon)^2B_\Gamma B^2_{\Gamma^T}+d^2(\epsilon^2+2B\epsilon)B_\Gamma^2](2B+\lambda+2\epsilon)\\
&<& 2\lambda B_{\Gamma,\Gamma^T}+9d^2\epsilon BB_{\Gamma,\Gamma^T}^2(3d^2\epsilon BB_{\Gamma,
\Gamma^T}+1) (2B+2\lambda).
\end{eqnarray*}
\end{proof}

\noindent{\bf Proof of Theorem \ref{thmomega}}\\
From Lemma \ref{concentration} we have that with probability greater than $1-p^{2-c}$, $\epsilon\leq \{(c\log p+\log C_1)/C_2n\}^{1/2}$. With some abuse of notations, we denote $\epsilon= \{(c\log p+\log C_1)/C_2n\}^{1/2}$. Choose
  \begin{eqnarray*}
 \lambda=\max\Big\{8\alpha^{-1}, \frac{3(2-\alpha)(2B+1)}{1-\alpha} d^2 BB_{\Gamma,\Gamma^T}[1+2(B_\Sigma^2+1/3)B_{\Gamma,\Gamma^T}]\Big\} \times
 \sqrt{\frac{c\log p +\log C_1}{C_2 n}},
 \end{eqnarray*}
 for some $c>2$ and since $d^2 B^2B_\Sigma^2B_{\Gamma,\Gamma^T}^2   
 \sqrt{\frac{\log p }{ n}}
 \rightarrow 0$ we can assume that the sample size $n$ is large enough such that
 \begin{eqnarray*}
 n&>&(c\log p+\log C_1) \times \max\{(C_2\min (B^2,1))^{-1},  81B^2d^4B_{\Gamma,\Gamma^T}^2C_2^{-1}, \\
 &&
 9(C_2\alpha)^{-1}B^2B^2_{\Gamma,\Gamma^T}[1+2(B_\Sigma^2+1/3)B_{\Gamma,\Gamma^T}]^2\}.
 \end{eqnarray*}
Clearly under the assumptions of Theorem \ref{thmomega} we have $\lambda=O\bigg( d^2 B^2B_\Sigma^2B_{\Gamma,\Gamma^T}^2   
\sqrt{\frac{\log p }{ n}}\bigg)$.
 We firstly  verify that the assumptions in Lemmas \ref{gammabound} and \ref{Omegalem} are true for the given $\lambda$, $n$ and $\epsilon$.

 (i) By noticing that $\frac{8\epsilon}{\alpha}<\lambda$, $\frac{\alpha}{8}< \frac{\alpha }{4(2-\alpha)}$ and $n>(C_2B^2)^{-1}(c\log p+\log C_1)$ we immediately have
\begin{eqnarray*}
\epsilon<\min\Big\{B,\frac{\alpha\lambda}{4(2-\alpha)}\Big\}<\min\Big\{B,\frac{\alpha\lambda}{2(2-\alpha)}\Big\}.
\end{eqnarray*}

(ii)
$n>81B^2d^4B_{\Gamma,\Gamma^T}^2(c\log p+\log C_1)/C_2$ implies
 $B_{\Gamma,\Gamma^T}<\frac{1}{9d^2B\epsilon}$. Together with $\epsilon<B$ from (i) we can see that Assumption (\ref{assumpBGamma}) holds.

(iii) Since $\epsilon<\frac{\alpha\lambda}{2(2-\alpha)}$, we have
\begin{eqnarray*}
C_\alpha>\frac{\alpha\lambda-4\epsilon }{ 2B \alpha\lambda+\alpha\lambda}>\frac{\alpha\lambda-2\alpha\lambda/(2-\alpha) }{ 2B \alpha\lambda+\alpha\lambda}
=\frac{1-\alpha} { (2-\alpha)(2B +1)}.
\end{eqnarray*}
Together with $B_{\Gamma,\Gamma^T}<\frac{1}{9d^2B\epsilon}$, we have
 \begin{eqnarray*}
&&3d^2\epsilon BB_{\Gamma,\Gamma^T} [1+(B_\Sigma^2+3d^2\epsilon B B_{\Gamma,\Gamma^T})(9d^2\epsilon B  B_{\Gamma,\Gamma^T}+1)B_{\Gamma,\Gamma^T}]\\
&<&3d^2\epsilon BB_{\Gamma,\Gamma^T}[1+2(B_\Sigma^2+1/3)B_{\Gamma,\Gamma^T}]\\
&\leq& \frac{1-\alpha} { (2-\alpha)(2B +1)} \lambda \\
&<&C_\alpha \lambda.
\end{eqnarray*}
(\ref{assumpLemma3}) is then true since $n> 9(C_2\alpha)^{-1}B^2B^2_{\Gamma,\Gamma^T}[1+2(B_\Sigma^2+1/3)B_{\Gamma,\Gamma^T}]^2(c\log p+\log C_1)$ implies
\begin{eqnarray*}
3\epsilon BB_{\Gamma,\Gamma^T} [1+(B_\Sigma^2+3d^2\epsilon B B_{\Gamma,\Gamma^T})(9d^2\epsilon B  B_{\Gamma,\Gamma^T}+1)B_{\Gamma,\Gamma^T}]\leq \alpha.
\end{eqnarray*}
(i), (ii) and (iii) and Lemma \ref{Omegalem} imply that
\begin{eqnarray*}
||\hat{\Omega}-\Omega||_\infty&<& 2\lambda B_{\Gamma,\Gamma^T}+9d^2\epsilon BB_{\Gamma,\Gamma^T}^2(3d^2\epsilon BB_{\Gamma,
\Gamma^T}+1) (2B+2\lambda)  \\
&\leq& 2\lambda B_{\Gamma,\Gamma^T}+12d^2\epsilon B B^2_{\Gamma,\Gamma^T}(2B+2\lambda)\\
&=&\frac{14}{3}\lambda B_{\Gamma,\Gamma^T}+ 24d^2\epsilon B^2 B^2_{\Gamma,\Gamma^T}.
\end{eqnarray*}

\subsection{Proofs of Theorem \ref{thmdelta}}

We first introduce some technical lemmas and the proof of Theorem \ref{thmdelta} will be given after these lemmas.

\begin{Lem}\label{sigmaD}
Assume that $A_\Sigma d_\delta \epsilon<1$ we have
\begin{eqnarray*}
||\hat{\Sigma}_{D,D}^{-1}-\Sigma_{D,D}^{-1}||_{1,\infty}&\leq& \frac{A_\Sigma^2d_\delta \epsilon}{1-A_\Sigma d_\delta \epsilon}.
\end{eqnarray*}
\end{Lem}
\begin{proof}
This lemma can be easily proved using the following observation:
\begin{eqnarray*}
||\hat{\Sigma}_{D,D}^{-1}-\Sigma_{D,D}^{-1}||_{1,\infty}&\leq&
||\hat{\Sigma}_{D,D}^{-1}||_{1,\infty} ||\hat{\Sigma}_{D,D}-\Sigma_{D,D} ||_{1,\infty}   ||{\Sigma}_{D,D}^{-1}||_{1,\infty} \\
&\leq& (A_\Sigma+||\hat{\Sigma}_{D,D}^{-1}-\Sigma_{D,D}^{-1}||_{1,\infty})d_\delta \epsilon A_\Sigma.
\end{eqnarray*}
\end{proof}

\begin{Lem}\label{gamma}
\begin{eqnarray*}
|\hat{\gamma}-\gamma|_\infty \leq 8\epsilon_\mu+ 2(\epsilon+B_\Sigma \epsilon_\Omega A_2+d\epsilon\epsilon_\Omega A_2) |\Omega (\mu_1-\mu_2)|_1
+2(B+\epsilon)(A_1+d\epsilon_\Omega)\epsilon_\mu.
\end{eqnarray*}
\end{Lem}

\begin{proof}
\begin{eqnarray*}\label{gamma1}
|\gamma-\hat{\gamma}|_\infty 
&\leq& 4|\hat{\Delta}_\mu-\Delta_{\mu}|_\infty+|(\Delta_1-\Delta_2)\Omega \Delta_\mu|_\infty+|(\hat{\Sigma}_1-\hat{\Sigma}_2)(\hat{\Omega}-\Omega)\Delta_\mu|_\infty\\
&&+|(\hat{\Sigma}_1-\hat{\Sigma}_2)\hat{\Omega}(\hat{\Delta}_\mu-\Delta_\mu)|_\infty.
\end{eqnarray*}
Lemma \ref{gamma} can then be proved using the following facts:
\begin{eqnarray*}
|\hat{\Delta}_\mu-\Delta_{\mu}|_\infty\leq 2\epsilon_\mu
\end{eqnarray*}
\begin{eqnarray*}\label{gamma2}
|\mu_1-\mu_2|_1\leq ||(\Omega^{-1})_{\cdot, D}||_{1,\infty}|\Omega (\mu_1-\mu_2)|_1=A_2|\Omega (\mu_1-\mu_2)|_1;
\end{eqnarray*}
\begin{eqnarray*}\label{gamma3}
|(\Delta_1-\Delta_2)\Omega \Delta_\mu|_\infty\leq 2\epsilon |\Omega (\mu_1-\mu_2)|_1;
\end{eqnarray*}
\begin{eqnarray*}\label{gamma4}
|(\hat{\Sigma}_1-\hat{\Sigma}_2)(\hat{\Omega}-\Omega)\Delta_\mu|_\infty &\leq& ||(\Sigma_1-\Sigma_2)(\hat{\Omega}-\Omega)||_\infty |\Delta_\mu|_1+
||(\Delta_1-\Delta_2)(\hat{\Omega}-\Omega)||_\infty|\Delta_\mu|_1 \nonumber\\
&\leq&2B_\Sigma\epsilon_\Omega |\mu_1-\mu_2|_1+2d\epsilon\epsilon_\Omega |\mu_1-\mu_2|_1;
\end{eqnarray*}
\begin{eqnarray*}\label{gamma5}
|(\hat{\Sigma}_1-\hat{\Sigma}_2)\hat{\Omega}(\hat{\Delta}_\mu-\Delta_\mu)|_\infty &\leq&
2(B +\epsilon)|\hat{\Omega}(\hat{\Delta}_\mu-\Delta_\mu)|_1   \nonumber \\
&\leq&2(B +\epsilon)[||{\Omega}||_{1,\infty}|\hat{\Delta}_\mu-\Delta_\mu|_\infty+d\epsilon_\Omega|\hat{\Delta}_\mu-\Delta_\mu|_\infty]  \nonumber \\
&\leq&2(B+\epsilon)(A_1+d\epsilon_\Omega)\epsilon_\mu.
\end{eqnarray*}

\end{proof}

\noindent
{\bf Proof of Theorem \ref{thmdelta}}
Using similar arguments as in Lemma \ref{concentration}, there exists a constant $C_\epsilon>0$ such that $\max\{\epsilon,\epsilon_\mu\}\leq  C_\epsilon\{(c\log p+\log C_1)/C_{2\delta}n\}^{1/2}$. Similar to the proof of Theorem \ref{thmomega}, we choose
 \begin{eqnarray*}
 \lambda_\delta&=&\max\Big\{ \frac{2(2-\alpha_\delta)C_\epsilon}{\alpha_\delta}[4+(2+B_\Sigma A_2)|\Omega(\mu_1-\mu_2)|_1+2B(A_1+C_3)] , \\ &&   \frac{d_\delta(5 A_\Sigma+2B_\Sigma A_\Sigma^2)}{C_\delta \alpha_\delta}  
 \times(C_3+1)\Big\}  \sqrt{\frac{c\log p +\log C_1}{C_{2\delta} n}},
 \end{eqnarray*}
where $C_3=\frac{14}{3}\max\Big\{\frac{8}{\alpha}, \frac{3(2-\alpha)(2B+1)}{1-\alpha}d^2  BB_{\Gamma,\Gamma^T}[1+2(B_\Sigma^2+\frac{1}{3})B_{\Gamma,\Gamma^T}]\Big\}B_{\Gamma,\Gamma^T}+24d^2B^2B_{\Gamma,\Gamma^T}^2$,
and assume that $n$ is large enough such that
 \begin{eqnarray*}
 n&>&(c\log p+\log C_1) \times \max\{C_{2\delta}^{-1}, 2C_2^{-1}A_\Sigma^2d_\delta^2, C_2^{-1}(A_2+1)^2d^2, \\ 
 &&C_2^{-1}C_\delta^{-2}\alpha_{\delta}^{-2}d_\delta^2(5 A_\Sigma+2B_\Sigma A_\Sigma^2)^2 \},
 \end{eqnarray*}
 where $0<C_\delta=\frac{\alpha_\delta\lambda_\delta-(2-\alpha_\delta)K\gamma}{\alpha_\delta\lambda_\delta(1+A_\gamma+K_\gamma)}<1$ and $C_{2\delta}=\min\{C_2,(2\sigma^2)^{-1}\}\times\min\{B^2,1\}$.

(i)
Suppose $\tilde{\delta}$ is the solution of:
\begin{eqnarray*}
\tilde{\delta}=\min_{\delta\in R^p,
\delta_{D^c}=0} \frac{1}{2} \delta^T (\hat\Sigma_1 +\hat\Sigma_2)\delta-\hat{\gamma}^T \delta +\lambda_\delta \| \delta\|_1,
\end{eqnarray*}
We first show that $\hat{\delta}=\tilde{\delta}$. It sufficies to show that for any $e\in D^c$,
\begin{eqnarray*}
|2\hat{\Sigma}_{e,D}\tilde{\delta}_{D}-\hat{\gamma}_e|\leq \lambda_\delta.
\end{eqnarray*}
By the definition of $\tilde{\delta}$ we have
\begin{eqnarray*}
\{2\hat{\Sigma}\tilde{\delta}-\hat{\gamma}+\lambda_\delta Z\}_D=\bf{0},
\end{eqnarray*}
where $Z=(Z_1,\ldots,Z_p)^T$ with $Z_i=0$ for $i\in D^c$, $Z_i={\rm sign}(\tilde{\delta})$ for $i\in D$ and $\tilde{\delta}\neq 0$, $Z_i\in[-1,1]$ for $i\in D$ and $\tilde{\delta}=0$. Consequently, we have $\tilde{\delta}_D=\frac{1}{2}\hat{\Sigma}_{D,D}^{-1}(\hat{\gamma}_D-\lambda_\delta Z_D)$. Together with the fact that $\Sigma_{e,D}\Sigma_{D,D}^{-1}\gamma_D=2\Sigma_{e,D}\Sigma_{D,D}^{-1} \Sigma_{D,D}\delta_D=\gamma_e$, we have
\begin{eqnarray}\label{delta1}
|2\hat{\Sigma}_{e,D}\tilde{\delta}_{D}-\hat{\gamma}_e|&=&
|\hat{\Sigma}_{e,D}\hat{\Sigma}_{D,D}^{-1}(\hat{\gamma}_D-\lambda_\delta Z_D)-\hat{\gamma}_e|   \nonumber \\
&\leq& |(\hat{\Sigma}_{e,D}\hat{\Sigma}_{D,D}^{-1}-\Sigma_{e,D}\Sigma_{D,D}^{-1}){\gamma}_D|+
|\hat{\Sigma}_{e,D}\hat{\Sigma}_{D,D}^{-1}(\hat{\gamma}_D-\gamma_D)|    \nonumber\\
&&+\lambda_\delta |\hat{\Sigma}_{e,D}\hat{\Sigma}_{D,D}^{-1}|_1+|\gamma_e-\hat{\gamma}_e|   \nonumber\\
&\leq& |(\hat{\Sigma}_{e,D}\hat{\Sigma}_{D,D}^{-1}-\Sigma_{e,D}\Sigma_{D,D}^{-1})|_1|{\gamma}_D|_\infty+\nonumber\\
&&
 |\hat{\Sigma}_{e,D}\hat{\Sigma}_{D,D}^{-1}|_1|(\hat{\gamma}_D-\gamma_D)|_\infty     +\lambda_\delta |\hat{\Sigma}_{e,D}\hat{\Sigma}_{D,D}^{-1}|_1+|\gamma_e-\hat{\gamma}_e|.
\end{eqnarray}
For simplicity, in the following, inequalities will be derived without mentioning whether they hold ``with probability greater than $1-p^{2-c}$''. For example,
since $n>2C_2^{-1}A_\Sigma^2d_\delta^2(c\log p+\log C_1),$ we have $A_\Sigma d_\delta\epsilon<1/2$ with probability greater than $1-p^{2-c}$ and we shall repeatedly use this inequality without mentioning it holds with probability greater than $1-p^{2-c}$. Since $n>C_{2\delta}^{-1}(c\log p+\log C_1)$, by (17) of \cite{Ravikumar:2011aa} and Theorem \ref{thmomega}, we also have
$\epsilon_\Omega\leq (C_3+1)\{(c\log p+\log C_1)/C_{2\delta}n\}^{1/2}:=\epsilon_0$. 

From Lemma \ref{sigmaD} we have,
\begin{eqnarray}\label{delta2}
&& |\hat{\Sigma}_{e,D}\hat{\Sigma}_{D,D}^{-1}-\Sigma_{e,D}\Sigma_{D,D}^{-1}|_1   \nonumber\\
&\leq& |(\hat{\Sigma}_{e,D}-\Sigma_{e,D})\Sigma_{D,D}^{-1}|_1+
           |\Sigma_{e,D}(\hat{\Sigma}_{D,D}^{-1}-\Sigma_{D,D}^{-1})|_1
 +|(\hat{\Sigma}_{e,D}-\Sigma_{e,D})(\hat{\Sigma}_{D,D}^{-1}-\Sigma_{D,D}^{-1})|_1 \nonumber \\
&\leq&   |\hat{\Sigma}_{e,D}-\Sigma_{e,D}|_1 |\Sigma_{D,D}^{-1}|_{1,\infty}+
          |\Sigma_{e,D}|_{1}||\hat{\Sigma}_{D,D}^{-1}-\Sigma_{D,D}^{-1}||_{1,\infty}  \nonumber\\
&&      +|\hat{\Sigma}_{e,D}-\Sigma_{e,D}|_1||\hat{\Sigma}_{D,D}^{-1}-\Sigma_{D,D}^{-1})||_{1,\infty} \nonumber \\
 &\leq&d_\delta\epsilon A_\Sigma+ d_\delta (B+\epsilon) ||\hat{\Sigma}_{D,D}^{-1}-\Sigma_{D,D}^{-1}||_{1,\infty}\nonumber\\
 &\leq& d_\delta\epsilon A_\Sigma+
  \frac{ d_\delta (B_\Sigma+d_\delta\epsilon) A_\Sigma^2  \epsilon}{1-A_\Sigma d_\delta \epsilon}.
\end{eqnarray}

Combining (\ref{delta1}), (\ref{delta2}) and Lemma \ref{gamma} we have:
\begin{eqnarray*}
&&|2\hat{\Sigma}_{e,D}\tilde{\delta}_{D}-\hat{\gamma}_e|   \\
&\leq& \left\{ \epsilon A_\Sigma +
  \frac{  (B+d_\delta\epsilon) A_\Sigma^2  \epsilon}{1-A_\Sigma d_\delta \epsilon} \right\}A_\gamma d_\delta+
  \left(2-\alpha_\delta+d_\delta\epsilon A_\Sigma+
  \frac{ d_\delta (B_\Sigma+d_\delta\epsilon) A_\Sigma^2  \epsilon}{1-A_\Sigma d_\delta \epsilon}\right)K_\gamma \\
&&+  \lambda_\delta \left(1-\alpha_\delta+d_\delta\epsilon A_\Sigma+
  \frac{ d_\delta (B_\Sigma+d_\delta\epsilon) A_\Sigma^2  \epsilon}{1-A_\Sigma d_\delta \epsilon}\right),
\end{eqnarray*}
where
\begin{eqnarray*}
K_\gamma= 2\epsilon_\mu[4+(2+B_\Sigma A_2)|\Omega(\mu_1-\mu_2)|_1+2B(A_1+C_3)],
\end{eqnarray*}
and we have used the fact that $\epsilon<B$, $d_\delta \epsilon A_2<1, d\epsilon<1$ and hence $8\epsilon_\mu+ 2(\epsilon+B_\Sigma \epsilon_\Omega A_2+d\epsilon\epsilon_\Omega A_2) |\Omega (\mu_1-\mu_2)|_1
+2(B+\epsilon)(A_1+d\epsilon_\Omega)\epsilon_\mu<K_\gamma$.
Assume that
\begin{eqnarray}\label{delta3}
d_\delta\epsilon A_\Sigma+
  \frac{ d_\delta (B_\Sigma+d_\delta\epsilon) A_\Sigma^2  \epsilon}{1-A_\Sigma d_\delta \epsilon}\leq C_\delta\alpha_\delta \min\{\lambda_\delta,1\},
\end{eqnarray}
where $C_\delta=\frac{\alpha_\delta\lambda_\delta-(2-\alpha_\delta)K\gamma}{\alpha_\delta\lambda_\delta(1+A_\gamma+K_\gamma)}$. It can be seen that $0<C_\delta<1$. We then have
\begin{eqnarray*}
|2\hat{\Sigma}_{e,D}\tilde{\delta}_{D}-\hat{\gamma}_e|
&\leq& \lambda_\delta A_\gamma C_\delta\alpha_\delta+(2-\alpha_\delta)K_\gamma+C_\delta\alpha_\delta \lambda_\delta K_\gamma+(1-\alpha_\delta+C_\delta\alpha_\delta)\lambda_\delta \\
&=&\lambda_\delta.
\end{eqnarray*}
Next we complete the Proof of this part by showing that (\ref{delta3}) holds.

Since $n>  C_2^{-1}C_\delta^{-2}\alpha_{\delta}^{-2}d^2_\delta(5 A_\Sigma+2B_\Sigma A_\Sigma^2)^2 (c\log p+\log C_1)$, we have
 \begin{eqnarray*}
d_\delta\epsilon A_\Sigma+
  \frac{ d_\delta (B_\Sigma+d_\delta\epsilon) A_\Sigma^2  \epsilon}{1-A_\Sigma d_\delta \epsilon}\leq
  d_\delta\epsilon A_\Sigma+
   { 4d_\delta A_\Sigma \epsilon}+2B_\Sigma A_\Sigma^2 d_\delta\epsilon  \leq
   C_\delta\alpha_\delta.
 \end{eqnarray*}
 On the other hand, since $\lambda_\delta\geq\frac{d_\delta(5 A_\Sigma+2B_\Sigma A_\Sigma^2)}{C_\delta \alpha_\delta}\epsilon_0,$ we have
 \begin{eqnarray*}
 d_\delta\epsilon A_\Sigma+
  \frac{ d_\delta (B_\Sigma+\epsilon d_\delta) A_\Sigma^2  \epsilon}{1-A_\Sigma d_\delta \epsilon}\leq \lambda_\delta   C_\delta\alpha_\delta.
 \end{eqnarray*}

(ii) Use the fact that $A_\Sigma d_\delta\epsilon<1/2$ we have:
\begin{eqnarray*}
|\hat{\delta}_D-\delta_D|_{\infty}&=&\frac{1}{2}|\hat{\Sigma}_{D,D}^{-1}(\hat{\gamma}_D-\lambda_\delta Z_D)-\Sigma^{-1}_{D,D}\gamma_D|\\
&\leq&|(\hat{\Sigma}_{D,D}^{-1}-\Sigma_{D,D}^{-1})\hat{\gamma}_D|_\infty+|\Sigma_{D,D}^{-1}(\hat{\gamma}_D-\gamma_D)|_\infty+\lambda_\delta |\hat{\Sigma}_{D,D}^{-1}|_{1,\infty}\\
&\leq&   \frac{A_\Sigma^2d_\delta \epsilon}{1-A_\Sigma d_\delta \epsilon} (A_\gamma+|\hat{\gamma}_D-\gamma_D|_\infty)+A_\Sigma|\hat{\gamma}_D-\gamma_D|_\infty \\&&+\lambda_\delta \left(A_\Sigma+ \frac{A_\Sigma^2d_\delta \epsilon}{1-A_\Sigma d_\delta \epsilon} \right)\\
&\leq&  \frac{A_\gamma A_\Sigma^2d_\delta \epsilon}{1-A_\Sigma d_\delta \epsilon} +K_\gamma \left[ \frac{A_\Sigma^2d_\delta \epsilon}{1-A_\Sigma d_\delta \epsilon}+A_\Sigma\right]\\&&+\lambda_\delta \left(A_\Sigma+ \frac{A_\Sigma^2d_\delta \epsilon}{1-A_\Sigma d_\delta \epsilon} \right) \\
&\leq& 2A_\gamma A_\Sigma^2d_\delta \epsilon+2K_\gamma A_\Sigma
+2\lambda_\delta A_\Sigma.
\end{eqnarray*}
This theorem is proved by plugging in $K_\gamma= 2\epsilon_\mu[4+(2+B_\Sigma A_2)|\Omega(\mu_1-\mu_2)|_1+2B(A_1+C_3)]$.

\subsection{Proofs of Theorem \ref{thm3}}

\begin{proof}
(i) With some abuse of notations we write $d(z)=(z-\mu)^T \Omega (z-\mu)+\delta^T (z-\mu)$ and
$\hat{d}(z)=(z-\hat{\mu})^T \hat{\Omega} (z-\hat{\mu})+\hat{\delta}^T (z-\hat{\mu})$.
\begin{eqnarray*}
R_n(1|2)&=&P(\hat{d}(z)+\eta>0 | z\sim N(\mu_2,\Sigma_2))\nonumber\\
&=&P(d(z)+\eta>d(z)-\hat{d}(z)| z\sim N(\mu_2,\Sigma_2)).
\end{eqnarray*}
Denote $z=(z_1,\ldots,z_p)^T$. Note that $z^T(\Omega-\hat{\Omega})z\leq \sum_{ij}|z_iz_j(\Omega-\hat{\Omega})_{ij}|$, by noticing that $E|z_iz_j|\leq (E z_i^2+Ez_j^2)/2\leq C_\mu^2+C_\Sigma$, we have $z^T(\Omega-\hat{\Omega})z=O_p(s||\Omega-\hat{\Omega}||_\infty)$. Using a similar argument for bounding $|z^T\Omega \mu-z\hat{\Omega}\hat{\mu}|$, $|\mu^T\Omega\mu-\hat{\mu}^T\hat{\Omega}\hat{\mu}|$, $(\delta-\hat{\delta})^Tz$ and $\delta^T\mu-\hat{\delta}^T\hat{\mu}$ we obtain,
\begin{eqnarray}\label{3.2}
&&|d(z)-\hat{d}(z)|\nonumber\\
&=&O_p\left( s d_0^2 B^2B_\Sigma^2B_{\Gamma,\Gamma^T}^2   
\sqrt{\frac{\log p }{ n}}+  d_0^2A_\Sigma^2    B^2B_\Sigma^3B_{\Gamma,\Gamma^T}^2  \sqrt{\frac{\log p  }{ n}} \right).
\end{eqnarray}
From Assumption 1 and (\ref{3.2}) and the mean value theorem, we have:
\begin{eqnarray*}
&&R_n(1|2)-R(1|2)\\&=&\int_{0}^{d(z)-\hat{d}(z)}F_2(z)dz\\
&=&O_p\left( s d_0^2 B^2B_\Sigma^2B_{\Gamma,\Gamma^T}^2   
\sqrt{\frac{\log p }{ n}}+  d_0^2A_\Sigma^2    B^2B_\Sigma^3B_{\Gamma,\Gamma^T}^2  \sqrt{\frac{\log p  }{ n}} \right).
\end{eqnarray*}
(i) is proved by noticing that the above equality is also true for $R_n(2|1)-R(2|1)$.

(ii)
Let $\Phi(\cdot)$ be the cumulative distribution function of a standard normal random variable. We have any constant $C_z>0$,
\begin{eqnarray*}
P(|z|_\infty>C_z\sqrt{\log p})\leq p\left[1-\Phi\left(\frac{C\sqrt{\log p}+C_\mu}{C_\Sigma^{1/2}}\right)\right].
\end{eqnarray*}
From Lemma 11 of \cite{Liu:2009aa} we have when $p$ is large enough, by choosing $C_z>\sqrt{2(c-1)}C_\Sigma^{1/2}$,
\begin{eqnarray*}
 p\left[1-\Phi\left(\frac{C_z\sqrt{\log p}+C_\mu}{C_\Sigma^{1/2}}\right)\right]\leq p^{2-c}.
\end{eqnarray*}
This together with Theorems \ref{thmomega} and \ref{thmdelta} and the proof in (i), we have with probability greater than $1-3p^{2-c}$,
\begin{eqnarray*}
&&|d(z)-\hat{d}(z)|\\
&=&O_p\left( s d_0^2 B^2B_\Sigma^2B_{\Gamma,\Gamma^T}^2 \log p  
\sqrt{\frac{\log p }{ n}}+  d_0^2A_\Sigma^2    B^2B_\Sigma^3B_{\Gamma,\Gamma^T}^2    \frac{\log p  }{ \sqrt{n}} \right).
\end{eqnarray*}
The rest of the proof is similar to that in the proof of (i).
\end{proof}

\subsection{Proof of Proposition \ref{etalemma1}}
\begin{proof}
Suppose $e_1> e_2\geq \eta$. Denote $c_i= |\Sigma_1 \Sigma_2^{-1}|^{1/2}\exp\{\frac{1}{2}[e_i-\frac{1}{4}(\mu_1-\mu_2)^T\Omega(\mu_1-\mu_2)]\}$ for $i=1,2$. We have,
\begin{eqnarray*}
R(d,e_i)&=&\pi_1\int_{D(z,e_1)<0}f_1(z)dz+\pi_2\int_{D(z,e_i)>0}f_2(z)dz\\
&=& \int_{f_1(z)/f_2(z)<c_i^{-1}}\pi_1f_1(z)dz+ \int_{f_1(z)/f_2(z)>c_i^{-1}}\pi_2f_2(z)dz.
\end{eqnarray*}
Since $e_1> e_2\geq \eta$, it can be easily shown that $c_1> c_2\geq \pi_1/\pi_2$.
Consequently we have
\begin{eqnarray*}
R(d,e_1)-R(d,e_2)=- \int_{c_1^{-1}< f_1(z)/f_2(z)<c_2^{-1}}[\pi_1f_1(z)-\pi_2f_2(z)]dz>0.
\end{eqnarray*}
Therefore, $R(d, e)$ is strictly monotone increasing on $e\in[\eta,\infty)$. The second statement can be similarly proved.
\end{proof}

\subsection{Proofs of Theorem \ref{thm4}}
We first introduce some technical lemmas and the proof of Theorem \ref{thm4} will be given after these lemmas.

For any constant $c$, define $l_c=\min\{{\essinf}_{z\in [-c, c]}F_i(z), i=1,2\}$.

\begin{Lem}\label{etalemma2}
For any constant $c>0$, we have for any $-c\leq \epsilon_\eta\leq c$, $R(d, \eta+\epsilon_\eta)-R(d, \eta)\leq \pi_2u_c|\epsilon_\eta|$ and $R(d, \eta+\epsilon_\eta)-R(d, \eta)\geq \epsilon_\eta^2\exp(-c/2)\pi_2l_c/4$.
\end{Lem}
\begin{proof}
Let's consider $0\leq\epsilon_\eta\leq c$ first.
Note that
\begin{eqnarray*}
R(d, \eta+\epsilon_\eta)-R(d, \eta)= \int_{-\epsilon_\eta<D(z,\eta)<0}[\pi_2f_2(z)-\pi_1f_1(z)]dz.
\end{eqnarray*}
We have
\begin{eqnarray*}
R(d, \eta+\epsilon_\eta)-R(d, \eta)\leq \int_{-\epsilon_\eta<D(z,\eta)<0}\pi_2f_2(z)dz=\pi_2\int_{-\epsilon_\eta}^0F_2(z)dz\leq \pi_2u_c\epsilon_\eta.
\end{eqnarray*}
By noticing that $1-\exp(-x/2)-x\exp(-c/2)/2$ is an increasing function in $[0,c]$ we have
\begin{eqnarray*}
R(d, \eta+\epsilon_\eta)-R(d, \eta)&\geq& \int_{-\epsilon_\eta<D(z,\eta)<-\epsilon_\eta/2}[\pi_2f_2(z)-\pi_1f_1(z)]dz \\
&\geq&\int_{-\epsilon_\eta<D(z,\eta)<-\epsilon_\eta/2}[1-\exp(-\epsilon_\eta/2)]\pi_2f_2(z)dz \\
&\geq&\int_{-\epsilon_\eta<D(z,\eta)<-\epsilon_\eta/2}\epsilon_\eta\exp(-c/2)\pi_2f_2(z)/2dz \\
&\geq&\epsilon_\eta^2\exp(-c/2)\pi_2l_c/4.
\end{eqnarray*}
The lemma is then proved using a same argument as above for $-c\leq\epsilon_\eta< 0$.
\end{proof}
\noindent
Clearly \ref{etalemma2} holds when $c$ is set to be $c_\eta$. Noted from the proof of Lemma \ref{etalemma2} that the bounds do not depend on $\eta$, we can claim that the bounds holds uniformly in $\epsilon_\eta\in [-c_\eta,c_\eta]$. Similarly, it can be shown that:
\begin{Lem}\label{etalemma3}
Suppose $\hat{d}-d=O_p(\Delta_d)$ with $\Delta_d\rightarrow 0$,  we have $R(\hat{d}, e)-R(d, e)= O_p(\Delta_du_c)$ uniformly in $e\in [-c_\eta, c_\eta]$.
\end{Lem}

\begin{Lem}\label{etalemma4}
Under the assumptions of Theorem \ref{thm3}, $R_n(\hat{d},e)\rightarrow R(d,e)$ in probability uniformly in $e\in [-c_\eta, c_\eta]$.
\end{Lem}
\begin{proof}
Denote all the samples in the two classes as $\{z_i, i=1,\ldots, n_1+n_2\}$ and denote the estimator obtained by leaving the $i$th sample out as $\hat{d}_{-i}$. Similarly we use $\hat{d}_{-(i,j)}$ to denote the estimator obtained by leaving the $i$th and $j$th samples out.
From (\ref{3.2}), we immediately have that for any $e\in[-c_\eta,c_\eta]$,
\begin{eqnarray*}
&&EI\{\hat{d}(z_i)+e>0\}-EI\{\hat{d}_{-i}(z_i)+e>0\}\\
&=&EI\{d(z_i)+e>0\}-EI\{d(z_i)+e>0\}+o(1)\\
&=&o(1).
\end{eqnarray*}
Together with Lemma \ref{etalemma3} we have
\begin{eqnarray}\label{biaseta}
ER_n(\hat{d},e)-R(d,e)=ER(\hat{d}_{-i},e)+o(1)-R(d,e)\rightarrow 0,
\end{eqnarray}
uniformly in $e\in [-c_\eta, c_\eta]$. Note that
\begin{eqnarray*}
Var(I\{\hat{d}(z_i)+e>0\})\leq \frac{1}{4},
\end{eqnarray*}
and for any $(i,j)\in \{(k,l): 1\leq k, l\leq n_1+n_2, i\neq j\}$,
\begin{eqnarray*}
&&Cov(I\{\hat{d}(z_i)+e>0\},I\{\hat{d}(z_j)+e>0\})\\
&=&Cov(I\{\hat{d}_{-(i,j)}(z_i)+e>0\},I\{\hat{d}_{-(i,j)}(z_j)+e>0\})+o(1)\\
&=&Cov(I\{{d}(z_i)+e>0\},I\{d(z_j)+e>0\})+o(1),
\end{eqnarray*}
where the last step can be obtained using (\ref{3.2}) and Lemma \ref{etalemma3} and the $o(1)$ term does not depend on $e$. Since $z_i, z_j$ are independent, we immediately have
\begin{eqnarray}\label{vareta}
Var(R_n(\hat{d},e))\rightarrow 0,
\end{eqnarray}
uniformly in $e\in[-c_\eta, c_\eta]$. The lemma is then proved by Markov's inequality and the uniform convergence of the bias (\ref{biaseta}) and the variance (\ref{vareta}) of $R_n(\hat{d},e)$.
\end{proof}

\noindent{\bf Proof of Theorem \ref{thm4}}\\
The result that $\hat{\eta}\rightarrow \eta$ can be obtained by Proposition \ref{etalemma1}, Lemma \ref{etalemma4} and Theorem 5.7 of \cite{Vardervart}. The second statement immediately follows from Theorem \ref{thm3}.


\vskip 0.2in

\renewcommand{\baselinestretch}{1}
\normalsize
\bibliographystyle{natbib}

\end{document}